\documentclass[11pt]{article}

\usepackage[T1]{fontenc}
\usepackage[utf8]{inputenc}
\usepackage{lmodern}
\usepackage[margin=1in]{geometry}
\usepackage{amsmath,amssymb,amsthm,mathtools}
\usepackage{booktabs}
\usepackage{tikz}
\usepackage{graphics,graphicx,color}
\usepackage[dvipsnames]{xcolor}
\usetikzlibrary{arrows.meta,positioning}
\usepackage{enumitem}
\usepackage{microtype}
\usepackage[hidelinks]{hyperref}
\usepackage[
  backend=biber,
  style=alphabetic,
  sorting=nyt,
  maxalphanames=4,
  minalphanames=4,
  maxbibnames=10
]{biblatex}
\usepackage{comment}
\usepackage{tcolorbox}

\usepackage{keytheorems}
\newkeytheoremstyle{boxthm}
  {
    noteseparator={ },
    notebraces={(}{)},
    headpunct={.},
    bodyfont=\normalfont,
    tcolorbox={
    coltitle=black,
    toptitle=2pt,
    colback=OliveGreen!5,
  colframe=OliveGreen!25,
  boxrule=1pt,
  arc=6pt,
  left=6pt,right=6pt,top=6pt,bottom=6pt
    }, 
  }
\newkeytheorem{protocol}[style=boxthm]

\newkeytheoremstyle{minorboxthm}
  {
    noteseparator={ },
    notebraces={(}{)},
    headpunct={.},
    bodyfont=\normalfont,
    tcolorbox={
    coltitle=black,
    toptitle=2pt,
    colback=Gray!5,
  colframe=Gray!25,
  boxrule=1pt,
  arc=6pt,
  left=6pt,right=6pt,top=6pt,bottom=6pt
    }, 
  }
\newkeytheorem{minorprotocol}[name=Reduction,style=minorboxthm,sibling=protocol]

\usepackage{xcolor}
\usepackage{fontawesome5}

\allowdisplaybreaks
\numberwithin{equation}{section}

\newtheorem{theorem}{Theorem}[section]
\newtheorem{proposition}[theorem]{Proposition}
\newtheorem{lemma}[theorem]{Lemma}
\newtheorem{corollary}[theorem]{Corollary}
\newtheorem{remark}[theorem]{Remark}
\newtheorem{definition}[theorem]{Definition}

\newcommand{\E}{\mathop{\mathbb{E}}}
\newcommand{\F}{\mathbb{F}}

\newcommand{\cK}{\mathcal{K}}
\newcommand{\Tr}{\operatorname{Tr}}
\newcommand{\Id}{I}
\newcommand{\ket}[1]{\lvert #1\rangle}
\newcommand{\bra}[1]{\langle #1\rvert}
\newcommand{\braket}[2]{\langle #1\mid #2\rangle}
\newcommand{\proj}[1]{\lvert #1\rangle\!\langle #1\rvert}
\newcommand{\GL}{\mathsf{GL}}
\newcommand{\CV}{\mathsf{CV}}
\newcommand{\Bop}{\mathsf{B}}
\newcommand{\Cop}{\mathsf{C}}
\newcommand{\Jop}{\mathsf{J}}
\newcommand{\RePart}{\operatorname{Re}}

\title{Unclonable encryption from BB84 states:\\
a simultaneous Goldreich-Levin reduction}
\author{
Andrea Coladangelo\thanks{University of Washington. Email: \texttt{coladan@cs.washington.edu}}
\quad
Qipeng Liu\thanks{University of California San Diego. Email: \texttt{qipengliu0@gmail.com}}
\quad
Ziyi Xie\thanks{Tsinghua University. Email: \texttt{xie-zy21@mails.tsinghua.edu.cn}}
}
\date{}

\begin{document}
\maketitle

\begin{abstract}
Goldreich-Levin reductions are ubiquitous in cryptography: they convert an algorithm capable of guessing $\langle r, m \rangle$ (mod $2$) for a hidden string $m$ and a random challenge $r$, to one that is capable of extracting the entirety of $m$.

Here, we describe a ``simultaneous'' Goldreich-Levin reduction for two entangled parties who are capable of guessing $\langle r, m \rangle$ given uniformly random \emph{identical} challenges $r$. This allows to upgrade \emph{any} unclonable encryption scheme satisfying ``search'' security to one satisfying the gold standard of unclonable ``indistinguishability''. As a corollary, we show that the simplest candidate unclonable encryption scheme from BB84 states satisfies unclonable indistinguishability.

This result was discovered by GPT-5.6 Ultra after a few interactions. Our prompts included recent results on unclonable encryption by Ananth and Sahai~\cite{AS26} and Ragavan~\cite{Rag26}. 
\end{abstract}

\medskip
\medskip
\medskip
\medskip
\begin{center}
\setlength{\fboxsep}{8pt}
\fbox{%
  \parbox{0.90\textwidth}{%
    \small
    \textbf{AI usage statement.}
    The result of this work was discovered by GPT-5.6 Ultra. In our prompts, we included the two recent results by Ananth and Sahai, and Ragavan, and we asked GPT to prove security of a scheme based solely on BB84 states (rather than Pauli states). GPT returned a correct proof. We then asked it to find a different proof based on a simultaneous Goldreich--Levin reduction, which it also found.

    \medskip
    We find the new reduction to be quite magical. We have spent considerable time trying to understand why it works and to arrive at a good conceptual explanation of it. As Terry Tao recently put it, echoing Bill Thurston, mathematics is not just about producing theorems and proofs, but about developing ``human understanding''. This experience has reinforced our belief that at least part of our role as researchers in the future will indeed be centered around developing human understanding of AI-discovered proofs. We hope that our community will embrace this role alongside the other roles that we hope researchers will continue to play. In this work, we try to do our part by developing what we hope is a satisfying conceptual understanding of the security of unclonable encryption.
  }%
}
\end{center}

\newpage
\hypertarget{mytoc}{}
\setcounter{tocdepth}{2}
\tableofcontents

\newpage

\section{Introduction}
Unclonable encryption provides a security property that no classical encryption scheme can achieve: an adversary may arbitrarily process a ciphertext into two parts \emph{before} the secret key is revealed, yet, two non-communicating recipients should not be able to simultaneously recover the plaintext, \emph{even after learning the secret key}~\cite{Got03,BL20}. 

There are two natural formulations of security for such an encryption scheme. In the \emph{search} formulation, modeled in Figure~\ref{fig:gl-games} below by a game $G$, the challenger samples a key-message pair $(k,m)$ from a distribution $\pi$, where $m\in \{0,1\}^n$ (for some security parameter $n$); sends a ciphertext or ``token'' state $\rho_{k,m}$ to $\mathcal{A}$, who processes the state into two parts; the two recipients, upon learning $k$, must both guess the \emph{entire} message $m$. Security requires their joint winning probability in $G$ to be negligible in the security parameter. This guarantee, however, does not rule out substantial shared leakage about the plaintext. For example, both recipients might learn its first bit while remaining unable to reconstruct the remaining bits. Yet, that single bit could already reveal the answer to a sensitive yes/no question encoded in the encrypted message.

The \emph{decision} formulation, or unclonable \emph{indistinguishability}, provides a stronger notion of security. Here, a hidden challenge bit selects one of two equal-length messages, and both recipients must identify which of the two messages was encrypted. A coordinated random guess trivially succeeds with probability $1/2$, and security requires that every attack achieves at most a negligible advantage over this baseline. In this sense, unclonable indistinguishability captures the idea that, after the ciphertext is processed or ``split'', at least one of the two recipients learns essentially nothing about the encrypted message. Establishing this stronger guarantee from simultaneous-search security is a central search-to-decision problem in unclonable cryptography.

The first and arguably simplest unclonable encryption scheme was proposed by Broadbent and Lord~\cite{BL20}. In this scheme, the ciphertext is a BB84 state: the encryption of a message $m\in \{0,1\}^n$ is $H^\theta\ket{m}=\bigotimes_{i=1}^n H^{\theta_i}\ket{m_i}$, where the basis string $\theta\in \{0,1\}^n$ is the secret key. This scheme satisfies search security, which follows directly from the analysis of the related BB84 monogamy-of-entanglement game by Tomamichel, Fehr, Kaniewski, and Wehner~\cite{TFKW13}.

A natural route to go from search security to decision security is suggested by the classic Goldreich--Levin theorem. Goldreich--Levin turns parity prediction into full recovery: an algorithm that predicts the (mod $2$) inner product $\langle r,m\rangle$ for a uniformly random mask $r\in \{0,1\}^n$ with noticeable advantage over $1/2$ can be transformed into one that recovers the hidden string $m$ with noticeable probability~\cite{GL89}. Thus, a random inner product, which can also be viewed as the parity of a random subset of the bits of $m$, is a ``hardcore bit'' of $m$: any noticeable ability to predict it can be leveraged to recover the entirety of $m$. It turns out that the Goldreich-Levin theorem is still valid, and is even more elegant, in the quantum world~\cite{AC02}: a single superposition query to the algorithm predicting the inner product suffices to recover the entirety of $m$!

At first sight, this is precisely the type of reduction needed to upgrade a search-secure unclonable encryption scheme to one satisfying the stronger \emph{indistinguishability} notion: given a search-secure scheme modeled by a game $G$, one can encrypt a bit $b$ by masking it with the inner product $\langle r,m\rangle$, i.e.\ using the latter as a one-time pad below). We denote by $\mathsf{GL}(G)$ the corresponding \emph{common-mask decision game} (see Figure~\ref{fig:gl-games} below): here, the challenger additionally samples a uniformly random mask $r$, and reveals $(k,r)$ to both parties; their goal is to both guess the bit $\langle r,m\rangle$. An attack on the one-bit encryption is equivalent to a strategy for $\mathsf{GL}(G)$. Thus, if one could apply Goldreich--Levin in this \emph{simultaneous} setting, such a strategy could then be converted into a strategy for the search game $G$ in which both recipients extract $m$, contradicting search security.

The subtlety is that the simultaneous setting requires a form of Goldreich--Levin that is not supplied by the standard reduction. In $\mathsf{GL}(G)$, after the split, the parties receive the original key $k$ as well as a \emph{common} random mask: the \emph{same} $r$ is revealed to both parties. 
The natural reduction in which Bob and Charlie each locally apply a standard quantum Goldreich-Levin reduction only works for \emph{independent} masks, i.e.\ when Bob and Charlie respectively receive some independently sampled $r$ and $s$, and are required to guess their respective inner products $\langle r, m\rangle$ and $\langle s, m\rangle$. This distinction is fundamental: earlier work on the feasibility of unclonable encryption already identified barriers showing that standard seeded-extractor and Goldreich--Levin compilers do not directly survive in the unclonable setting~\cite{AKLLZ22}.

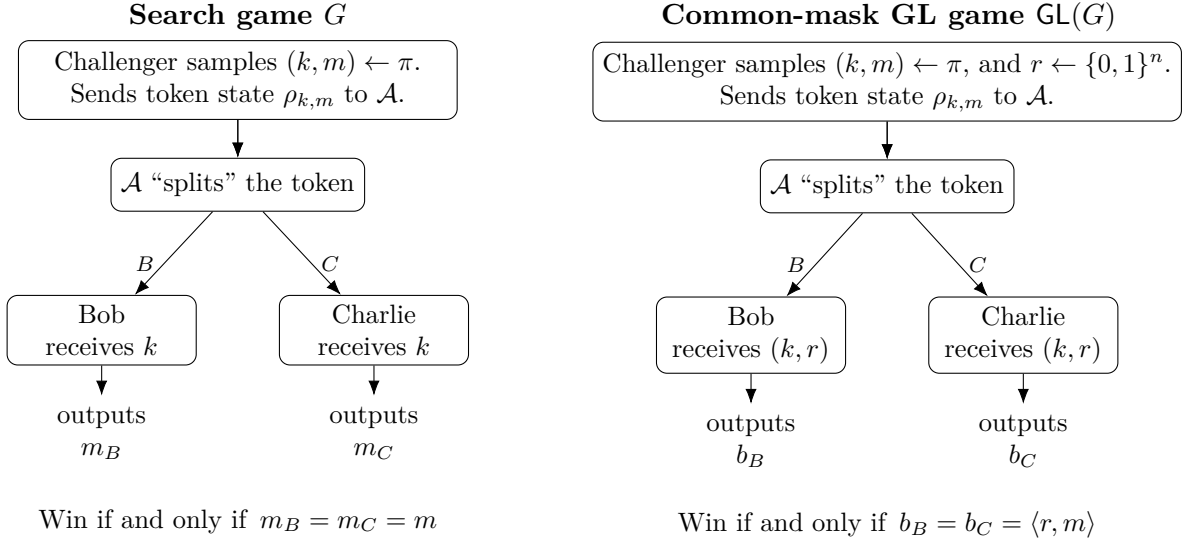
\begin{figure}
\begin{minipage}[t]{0.48\textwidth}
\centering
\textbf{Search game $G$}\par\smallskip
\begin{tikzpicture}[
  >={Latex[length=2mm]},
  every node/.style={font=\small},
  box/.style={draw,rounded corners,align=center,minimum height=7mm},
  party/.style={box,minimum width=25mm}
]
  \node[box,minimum width=58mm,align=center] (state)
  {Challenger samples $(k,m) \leftarrow \pi$.\\
   Sends token state $\rho_{k,m}$ to $\mathcal A$.};

\node[box,below=5mm of state,minimum width=32mm] (alice)
  {$\mathcal A$ ``splits'' the token};

\draw[->] (state) -- (alice);

\node[party,below=11mm of alice,xshift=-18mm] (bob)
  {Bob\\receives $k$};
\node[party,below=11mm of alice,xshift=18mm] (charlie)
  {Charlie\\receives $k$};
\node[below=4mm of bob,align=center] (bout)
  {outputs\\$m_B$};
\node[below=4mm of charlie,align=center] (cout)
  {outputs\\$m_C$};

  \draw[->] (state) -- (alice);
  \draw[->] (alice) -- node[pos=0.62,left,font=\scriptsize] {$B$} (bob);
  \draw[->] (alice) -- node[pos=0.62,right,font=\scriptsize] {$C$} (charlie);
  \draw[->] (bob) -- (bout);
  \draw[->] (charlie) -- (cout);

  \node[below=38mm of alice,align=center]
    {$\mathrm{Win} \ \text{if and only if }  \, m_B=m_C=m$};
\end{tikzpicture}
\end{minipage}\hfill
\begin{minipage}[t]{0.48\textwidth}
\centering
\textbf{Common-mask GL game $\mathsf{GL}(G)$}\par\smallskip
\begin{tikzpicture}[
  >={Latex[length=2mm]},
  every node/.style={font=\small},
  box/.style={draw,rounded corners,align=center,minimum height=7mm},
  party/.style={box,minimum width=25mm}
]
   \node[box,minimum width=58mm,align=center] (state)
  {Challenger samples $(k,m) \leftarrow \pi$, and $r \leftarrow \{0,1\}^n$.\\
   Sends token state $\rho_{k,m}$ to $\mathcal A$.};

\node[box,below=5mm of state,minimum width=32mm] (alice)
  {$\mathcal A$ ``splits'' the token};

\draw[->] (state) -- (alice);
  \node[party,below=11mm of alice,xshift=-18mm] (bob)
    {Bob\\receives $(k,r)$};
  \node[party,below=11mm of alice,xshift=18mm] (charlie)
    {Charlie\\receives $(k,r)$};
  \node[below=4mm of bob,align=center] (bout)
    {outputs\\$b_B$};
  \node[below=4mm of charlie,align=center] (cout)
    {outputs\\$b_C$};

  \draw[->] (state) -- (alice);
  \draw[->] (alice) -- node[pos=0.62,left,font=\scriptsize] {$B$} (bob);
  \draw[->] (alice) -- node[pos=0.62,right,font=\scriptsize] {$C$} (charlie);
  \draw[->] (bob) -- (bout);
  \draw[->] (charlie) -- (cout);

  \node[below=38mm of alice,align=center]
    {$\mathrm{Win}\ \text{if and only if } \,    b_B=b_C=\langle r,m\rangle$};
\end{tikzpicture}
\end{minipage}
\caption{The search game $G$ (left) and its common-mask
Goldreich--Levin variant $\mathsf{GL}(G)$ (right).}
\label{fig:gl-games}
\end{figure}

\paragraph{Simultaneous Goldreich--Levin and the common-mask barrier.}
The single-party quantum Goldreich--Levin reduction was developed by
Adcock and Cleve~\cite{AC02}. Later, Coladangelo, Liu, Liu, and Zhandry
established a version allowing quantum auxiliary input, which is useful in other settings of unclonable cryptography~\cite{CLLZ21}. Kundu and Tan and,
independently, Ananth, Kaleoglu, and Liu extended this approach to the setting
of two non-communicating parties~\cite{KT25,AKL23}. Crucially, however,
these simultaneous reductions use \emph{independent} random masks $r$ and $s$
for the two parties. 

Extending these results to the \emph{common} mask setting has proven substantially more
elusive. Kundu and Tan explicitly observe that their techniques do not extend
to the case in which both recipients receive the same mask~\cite{KT25}.
Ananth and Behera formulate related identical-challenge simultaneous
inner-product statements as conjectures~\cite{AB24}, while Chevalier,
Hermouet, and Vu propose closely related conjectures in the context of
simultaneous compute-and-compare obfuscation~\cite{CHV24}. Most directly,
Ananth, Kaleoglu, and Yuen identify a common-mask simultaneous
Goldreich--Levin theorem as precisely the missing ingredient needed for their
classical-key Goldreich--Levin compiler~\cite{AKY25}. The importance of this distinction becomes concrete for the BB84
monogamy game. Coladangelo, Liu, and Xie~\cite{CLX26} studied two decision
variants of this game and showed that merely asking both players to predict a
fixed parity does not suffice: a fixed parity can be guessed with constant
bias for every $n$, so ordinary XOR repetition fails to amplify security. By
contrast, when a uniformly random mask $r$ is revealed identically to both
players, they proved exponentially small advantage for semi-classical
strategies and conjectured the same behavior for arbitrary quantum
strategies. Related ``no-telegraphing'' results establish security in similar models~\cite{CKNY25,CakanGoyal24}, but also do not
yield the fully quantum common-mask search-to-decision reduction needed here.

Thus, in this work, we focus on the following question:
\begin{center}
\emph{Does there exist a simultaneous Goldreich-Levin reduction \\for two entangled quantum parties in the common-mask setting?\\} 
\end{center}

\paragraph{Recent momentous progress on unclonable encryption.}
In the meantime, there has been rapid progress via other methods recently culminating in a proof of the unconditional existence of unclonable encryption satisfying the gold standard of indistinguishable-security for single-bit messages. 
First, Bhattacharyya and Culf obtained inverse-polynomial
security~\cite{BC26}. Then, in a breakthrough work, Bhattacharyya, Broadbent, and Culf obtained negligible
security, albeit with inefficient encryption and decryption~\cite{BBC26}.
Two independent and concurrent AI-assisted breakthroughs of Ananth and Sahai~\cite{AS26}, and Ragavan~\cite{Rag26} proved the existence of efficient constructions from Pauli eigenstates with negligible security (schemes based on Pauli eigenstates and the general notion of Pauli encodings were also considered in \cite{BotteronEtAl26} and \cite{BDF26}, though negligible security is not proven there). Finally, Bhattacharyya, Broadbent, and Culf~\cite{bhattacharyya2026statistically} extended these results to unclonable encryption of arbitrary messages.

However, all of these results rely on direct, quite technical, analysis tailored to each of their proposed schemes, leaving open the elusive question of a general reduction from $\mathsf{GL}(G)$ to $G$. There are several appealing reasons for wanting such a reduction:
\begin{itemize}
\item The security of the simplest conjectured unclonable encryption scheme based on BB84 states~\cite{BL20, CLX26} is still open. A simultaneous Goldreich-Levin reduction would immediately prove its security.
\item Such a reduction would establish that, for unclonable encryption, search security can be generically upgraded to the stronger indistinguishability guarantee. This would be a clean, modular bridge from search to decision, avoiding the need for a separate analysis for each unclonable encryption construction. In particular, \emph{any} scheme with a search security game $G$ would automatically inherit security for the Goldreich--Levin compiled game $\mathsf{GL}(G)$ from a bound on the optimal value in $G$, which is often considerably simpler to establish directly. 
\item Finally, given the remarkable versatility of the classical Goldreich--Levin theorem, we expect that a common-mask simultaneous Goldreich--Levin reduction would find applications beyond unclonable encryption, as a useful tool for constructing and analyzing other unclonable cryptographic primitives. For example, such a reduction (if efficient) would also allow one to lift search security to decision security in the \emph{computational} setting, where a direct information-theoretic analysis, as in the recent sequence of papers~\cite{BBC26, AS26, Rag26}, may be hard to carry out.

\end{itemize}
\subsection{Our result}

We prove that a common-mask Goldreich--Levin reduction exists. Our
theorem is information-theoretic and is not tied to BB84 or any other specific construction. Let $G$ be any search unclonable-encryption game, and let $\mathsf{GL}(G)$ be its corresponding common-mask GL game (as in Figure~\ref{fig:gl-games}). We show the following.  
\begin{theorem}[Informal]\label{thm:intro-main}
Let $p_{\mathrm{Search}}$ and $p_{\mathrm{Pred}}$ be the optimal success probabilities in $G$ and $\mathsf{GL}(G)$, respectively, and define
\begin{equation*}
  \Delta:=2p_{\mathrm{Pred}}-1.
\end{equation*}
Then,
\begin{equation*}
  p_{\mathrm{Search}}
  \geq
  \Delta^4 \,.
\end{equation*}
\end{theorem}
More concretely, there is an explicit reduction that transforms any strategy $\mathcal P$ for $\mathsf{GL}(G)$, winning with probability $\frac12 + \frac{\Delta}{2}$ into a strategy $\mathcal P'$ for $G$ winning with probability at least $\Delta^4$. Thus, as long as $\Delta$ is non-negligible, the reduction yields a non-negligible success probability in $G$.

While the complexity of the reduction is not important when adversaries are unbounded, we note that our reduction is efficient in the following sense: if a lower bound on $0<\gamma \leq \Delta$ is known, the reduction is uniform and has running time polynomial in the running time of $\mathcal{P}$, and $1/\gamma$. When a lower bound on $\Delta$ is not known, we show that there is still a uniform reduction with \emph{expected} running time linear in the running time of $\mathcal{P}$, but at the cost of a polynomial loss in the search success guarantee, which becomes $p_{\mathrm{Search}}\geq \Omega(\Delta^6)$.

Our reduction is rooted in some ideas already present in the security analysis of Ananth and Sahai~\cite{AS26}. In particular, one key idea of considering a certain \emph{filter} to reduce Bob and Charlie's ``disagreement'' appears in the security analysis of Ananth and Sahai, though it does not have any operational meaning there. The main new insight is that a version of that filter has an operational meaning, reminiscent to that of an \emph{approximate ground state projector}~\cite{arad2012improved, anshu2022area}, and can be actually implemented algorithmically to achieve a Goldreich-Levin reduction that applies to \emph{any} unclonable encryption game.

While in this paper we focus on unclonable-encryption games as described above, we note that our reduction from Theorem~\ref{thm:intro-main} applies also to \emph{monogamy games}, essentially unchanged. We discuss this briefly in Remark~\ref{rem:1}.

Our Theorem~\ref{thm:intro-main} immediately yields two new constructions of unclonable encryption based on BB84 states~\cite{TFKW13} and coset states~\cite{CLLZ21,CV22}.
\begin{corollary}[Unclonable encryption from BB84 states]
\label{cor:intro-bb84-ue}
Consider the following one-bit encryption scheme. To encrypt a bit $b$, sample
$\theta,r,m\gets \{0,1\}^n$, use $(\theta,r)$ as the secret key, and output
the ciphertext
\[
\left(
H^\theta\ket m,
b\oplus \langle r, m \rangle
\right).
\]
Decryption measures the quantum register in the basis specified by $\theta$
to recover $m$ and then removes the $\langle r, m\rangle$ one-time pad.
This scheme satisfies unclonable
indistinguishability.  More precisely, the
optimal winning probability in $\mathsf{GL}(G)$ (where $G$ is the corresponding search game) is at most
\begin{equation*}
  \frac12+\frac12\cos^{n/2}\!\left(\frac{\pi}{8}\right).
\end{equation*}
\end{corollary}
The latter follows immediately from Theorem~\ref{thm:intro-main} along with security of the corresponding search game from \cite{TFKW13}, which says that the optimal winning probability in the search game is $\cos^{2n}\left(\frac{\pi}{8}\right)$. This also resolves a conjecture posed in \cite{CLX26} (there, the conjecture is posed for monogamy games, which our reduction also applies to).

\begin{corollary}[Unclonable encryption from subspace coset states]
\label{cor:intro-coset-ue}
Let $n$ be even.  Sample a uniformly random $n/2$-dimensional subspace
$A\subseteq\F_2^n$, a uniform $r\in\{0,1\}^n$, and uniform
$s,s'\in\F_2^n$.  For each $A$, fix an efficiently computable bijection
\[
  m_A:
  (\F_2^n/A)\times(\F_2^n/A^\perp)
  \longrightarrow \{0,1\}^n,
\]
where $\F_2^n/A$ and $\F_2^n/A^\perp$ are the spaces of cosets of $A$ and $A^\perp$ respectively, and abbreviate $m_A(s,s'):=m_A(s+A,s'+A^\perp)$.  Encrypt a bit $b$ using
secret key $(A,r)$ as
\[
  \left(
    \ket{A_{s,s'}},
    b\oplus\langle r,m_A(s,s')\rangle
  \right),
  \qquad
  \ket{A_{s,s'}}:=X^sZ^{s'}\ket A,
  \qquad
  \ket A:=\frac1{\sqrt{|A|}}\sum_{a\in A}\ket a.
\]
The scheme satisfies unclonable indistinguishability. 
The optimal winning probability in $\mathsf{GL}(G)$ (where $G$ is the corresponding search game) is at most
\begin{equation*}
  \frac12
  +
  \frac{e^{1/8}}2
  \cos^{n/4}\!\left(\frac{\pi}{8}\right).
\end{equation*}
Moreover, assuming post-quantum indistinguishability obfuscation and one-way
functions, the ciphertext may additionally contain obfuscated membership
programs for $A+s$ and $A^\perp+s'$ while retaining computational unclonable
indistinguishability against quantum polynomial-time adversaries. 
\end{corollary}
The latter follows from \cite{CV22}, which shows that the optimal winning probability in the corresponding search game is $\sqrt{e}\cos^{n}\!\left(\frac{\pi}{8}\right)$.

We remark that the upper bounds in
Corollaries~\ref{cor:intro-bb84-ue}
and~\ref{cor:intro-coset-ue} are most likely not tight, as they are obtained by combining our general search-to-decision reduction, which pays a quartic loss, with the corresponding search bounds. In particular, numerical evidence strongly suggests that the tight upper bound for unclonable encryption from BB84 states is actually $\frac12+\frac12\cos^{2n}\!\left(\frac{\pi}{8}\right)$.

Finally, we note the concurrent and independent work of Ananth~\cite{ananth2026non}, which appeared just a day after the first version of our work. This also establishes a search-to-decision theorem for the common-mask setting. Both works show that predicting a common inner product implies a non-negligible probability of simultaneously recovering the underlying string. Quantitatively, our result gives the stronger bound $p_{\mathrm{search}} \geq \Delta^4$, whereas Ananth proves $p_{\mathrm{search}} \geq \left(\frac{\Delta}{10}\right)^{22}$.
There is also a distinction in efficiency: our proof gives an explicit efficient reduction from a prediction strategy to a search strategy (and hence applies directly in computational settings), whereas Ananth's result is information-theoretic and does not provide such an efficient reduction. The two works use substantially different techniques. Ananth additionally develops extensions beyond the binary setting to inner products over arbitrary finite fields.

\section{Technical overview}
\label{sec:tech-overview}
\paragraph{The goal and notation.}
Fix a search game $G$ as in Figure~\ref{fig:gl-games} (a more detailed description is also given in Section~\ref{sec:model}). We use the same notation: $k$ is the key, $m\in\{0,1\}^n$ is the hidden string, $\rho^A_{k,m}$ is the token state, which is split by $\mathcal A$ into $\ket{\psi_{k,m}}^{BC}$ (here the superscripts indicate registers, where $B$ and $C$ can be of arbitrary dimension). We take the state after the split to be pure without loss of generality since $\mathcal{A}$ can always supply any purifying register to one of the two parties (and this can only increase their success probability). Recall that $\mathcal{A}$ does not have the key $k$, which is only revealed to Bob and Charlie, who cannot communicate. Their goal in $G$ is to guess $m$. In $\mathsf{GL}(G)$, Bob and Charlie additionally receive the same mask
$r\in\{0,1\}^n$, and their goal is to guess the bit $\langle r,m\rangle
  :=
  r_1m_1\oplus\cdots\oplus r_nm_n$.

Fix a strategy $\mathcal P$ for $\mathsf{GL}(G)$, and let $p_{\mathrm{Pred}}$ be its success probability. Set
\[
  \Delta:=2p_{\mathrm{Pred}}-1.
\]
Our reduction keeps the splitting procedure $\mathcal A$ unchanged and constructs
a strategy for $G$ where the probability of jointly recovering $m$ is at least $\Delta^4$.

\paragraph{Warmup: The single-party quantum Goldreich--Levin.}
We start by recalling the single-party quantum Goldreich--Levin reduction from Adcock and Cleve~\cite{AC02}. For a key $k$ and a mask $r$, let Bob's local strategy in $\mathsf{GL}(G)$ be described by projectors $\Pi^B_{k,r,0}$ and $\Pi^B_{k,r,1}$. We will refer to this as Bob's \emph{decoder} for simplicity. Equivalently, Bob's decoder in $\mathsf{GL}(G)$ consists of measuring the binary observable
\[
  V^B_{k,r}:=\Pi^B_{k,r,0}-\Pi^B_{k,r,1}.
\]
Now, notice that $V^B_{k,r}$ is nothing other than a \emph{reflection} (and hence a unitary). So, we can use it to define the following controlled-unitary. Let $R$ be an $n$-qubit auxiliary register. Define the ``mask-controlled'' reflection
\[
  \CV^B_k
  :=
  \sum_{r\in\{0,1\}^n}
  \bigl(\ket r\!\bra r\bigr)^R\otimes V^B_{k,r}\,,
\]
which acts coherently as
$\sum_r\alpha_r\ket r^R\ket\phi^B
  \ \longmapsto\
  \sum_r\alpha_r\ket r^R V^B_{k,r}\ket\phi^B$.
  
The GL reduction circuit of \cite{AC02} places $R$ in a uniform superposition, applies this
controlled reflection, and ``recombines'' the $R$ branches with a second
Hadamard transform:
\[
  \GL^B_k
  :=
  (H_R^{\otimes n}\otimes\Id_B)\,
  \CV^B_k\,
  (H_R^{\otimes n}\otimes\Id_B).
\]
For an arbitrary state $\ket\phi^B$, expanding these three operations gives
\begin{equation}
  \GL^B_k\bigl(\ket{0^n}^R\ket\phi^B\bigr)
  =
  \frac1{2^n}\sum_{y,r}
  (-1)^{\langle r,y\rangle}
  \ket y^R V^B_{k,r}\ket\phi^B.
  \label{eq:overview-basic-GL}
\end{equation}
We define $V^C_{k,r}$ and $\GL^C_k$ analogously for Charlie, introducing a separate auxiliary register $S$ (though we will not need Charlie just yet in the single-party setting). We denote the operator-valued Fourier coefficients of the two reflection families by
\begin{equation*}
  \widehat V^B_k(y)
  :=
  \frac1{2^n}\sum_r
  (-1)^{\langle r,y\rangle}V^B_{k,r},
  \qquad
  \widehat V^C_k(y)
  :=
  \frac1{2^n}\sum_s
  (-1)^{\langle s,y\rangle}V^C_{k,s}.
\end{equation*}
Then, Equation~\eqref{eq:overview-basic-GL} can be written more compactly as
\[
  \GL^B_k\bigl(\ket{0^n}^R\ket\phi^B\bigr)
  =
  \sum_y\ket y^R\widehat V^B_k(y)\ket\phi^B,
\]
and similarly for Charlie. Note then that the probability that Bob's GL ``extractor'' outputs $m$ is simply $\| \widehat V^B_k(m)\ket\phi^B  \|^2$.

Why does the reduction work in the single-party setting? The controlled reflection \\$\sum_{r\in\{0,1\}^n}
  \bigl(\ket r\!\bra r\bigr)^R\otimes V^B_{k,r}$ coherently converts Bob's guess for $\langle r,m\rangle$ into a phase on the branch indexed by $r$. If the decoder is usually correct, this phase is usually equal to $(-1)^{\langle r,m\rangle}$. The final Hadamard transform then cancels this phase precisely on the outcome $m$, causing the corresponding branches to interfere constructively. Equivalently, the correlation with the hidden inner product translates into weight on the branch $\widehat{V}^B_k(m)\ket\phi^B$ corresponding to the Fourier coefficient at $m$. This is the main idea behind the single-party quantum Goldreich--Levin reduction, and formalizing it is not difficult (we refer the reader, for example, to Appendix B.3 in \cite{CLLZ21} for the details).

\paragraph{Independent-mask simultaneous GL.}
We now consider the \emph{independent}-mask analogue of $\mathsf{GL}(G)$, in which Bob and Charlie receive independent uniformly random masks $r$ and $s$ respectively. In this setting, the simultaneous GL reduction is the most natural quite (and fairly straightforward to analyze): Bob and Charlie each locally apply their respective GL extractor, as described above.

Recall that we denote by $\ket{\psi_{k,m}}^{BC}$ the state of Bob and Charlie after the split. In the reduction, they apply $\GL^B_k$ and $\GL^C_k$ to
their respective registers $RB$ and $SC$. Then, Equation~\eqref{eq:overview-basic-GL}, applied jointly to registers $RB$ and $SC$ implies that the unnormalized leftover state on $BC$, conditioned on both parties outputting $m$ (i.e.\ both registers $R$ and $S$ containing $m$) is
\begin{equation}
  \begin{aligned}
  \ket{\eta^{\mathrm{ind}}_{k,m}}
  &:=
  \bigl(
    \widehat V^B_k(m)\otimes
    \widehat V^C_k(m)
  \bigr)\ket{\psi_{k,m}}
  \\
  &=
  \E_{\substack{r,s\gets\{0,1\}^n}}\!\left[
    (-1)^{\langle r,m\rangle}(-1)^{\langle s,m\rangle}
    \bigl(V^B_{k,r}\otimes V^C_{k,s}\bigr)
  \right]
  \ket{\psi_{k,m}}.
  \end{aligned}
  \label{eq:overview-independent-branch}
\end{equation}
Thus, the probability of simultaneous correct recovery of $m$ is precisely $\| \eta^{\mathrm{ind}}_{k,m}\|^2$.

It is useful now to separate the two individual correct-output
branches. Define 
\[
  \ket{b_{k,m}}:=\widehat V^B_k(m)\otimes\Id_C\ket{\psi_{k,m}},
  \qquad
  \ket{c_{k,m}}:=\Id_B\otimes\widehat V^C_k(m)\ket{\psi_{k,m}}.
\]
The vector $\ket{b_{k,m}}$ is the unnormalized leftover state when Bob alone
runs his GL circuit and obtains $m$; $\ket{c_{k,m}}$ has the analogous meaning
when Charlie alone runs his circuit.  Thus, $\|\ket{b_{k,m}}\|^2$ and
$\|\ket{c_{k,m}}\|^2$ describe the two individual correct-output probabilities.
Simultaneous recovery, instead, is captured by the branch $\ket{\eta^{\mathrm{ind}}_{k,m}} = \bigl(
    \widehat V^B_k(m)\otimes
    \widehat V^C_k(m)
  \bigr)\ket{\psi_{k,m}}$. Now, applying Cauchy--Schwarz gives
\begin{equation}
\label{eq:100}
  \bigl\|\eta^{\mathrm{ind}}_{k,m}\bigr\|^2
  \geq
  \left|
    \braket{\psi_{k,m}}{\eta^{\mathrm{ind}}_{k,m}}
  \right|^2
  =
  \left|\braket{b_{k,m}}{c_{k,m}}\right|^2.
\end{equation}
Thus, we can lower bound the joint success probability by the \emph{overlap between the individual correct-output branches} of Bob and Charlie. Crucially, large norms of these two vectors alone do not guarantee
a large joint success probability: the two vectors must also \emph{align} well. In the independent-mask setting, however, it turns out that successful joint prediction, on average over the independent challenges $r$ and $s$, guarantees precisely the required alignment. Indeed, expanding the overlap and the Fourier expressions gives
\[
\braket{b_{k,m}}{c_{k,m}}
=
\E_{r,s\gets\{0,1\}^n}
\left[
(-1)^{\langle r,m\rangle+\langle s,m\rangle}
\bra{\psi_{k,m}}
\bigl(V^B_{k,r}\otimes V^C_{k,s}\bigr)
\ket{\psi_{k,m}}
\right].
\]
Crucially, the independent-mask experiment samples exactly the same pairs $(r,s)$ that occur in the latter overlap! This is what allows us to bound the RHS in terms of the probability that Bob and Charlie are simultaneously correct. To see this, note that, for fixed $(k,m,r,s)$, the expression inside the expectation is the ``correlation'' between Bob's and Charlie's prediction outcomes after each outcome is multiplied by the sign corresponding to the correct parity. This correlation is $+1$ when Bob and Charlie are either both correct or both incorrect, and $-1$ when exactly one of them is correct. Letting $p_{kmrs}$ denote the probability that both players are correct, the latter correlation is therefore at least $p_{kmrs} - (1-p_{kmrs}) = 2 p_{kmrs} -1$.

Thus, letting $p_{\mathrm{Pred}}^{\mathrm{ind}}\geq1/2$ denote the probability that both players are correct, averaged over $k,m,r,s$, gives
\begin{equation}
  \E_{(k,m)\leftarrow\pi}
  \braket{b_{k,m}}{c_{k,m}}
  \geq
  2p_{\mathrm{Pred}}^{\mathrm{ind}}-1.
  \label{eq:overview-independent-correlation}
\end{equation} 
Combining Equation~\eqref{eq:overview-independent-correlation} with \eqref{eq:100} and Jensen's inequality gives the desired lower bound on the probability of successful extraction, $p_{\mathrm{Search}}$, and thus on the optimal success probability in $G$:
\[
  p_{\mathrm{Search}}
  \geq
  \bigl(2p_{\mathrm{Pred}}^{\mathrm{ind}}-1\bigr)^2 = \Delta^2 \,.
\]
The full calculation
appears in the proof of Proposition~\ref{prop:independent-mask-reduction}. 

\paragraph{The common-mask barrier.}
We are now ready to move to $\mathsf{GL}(G)$, the game that is relevant to unclonable encryption. Here, the crucial difference from earlier is that the \emph{same} mask $r$ is given to both parties.  Using the same notation as before, the winning probability in $\mathsf{GL}(G)$, i.e.\ the probability that both Bob and Charlie predict correctly, is
\begin{equation}
  p_{\mathrm{Pred}}
  =
  \E_{\substack{(k,m)\leftarrow\pi\\
                 r\gets\{0,1\}^n}}\!\left[
  \bra{\psi_{k,m}}
  \left(
    \frac{\Id_B+(-1)^{\langle r,m\rangle}V^B_{k,r}}2
    \otimes
    \frac{\Id_C+(-1)^{\langle r,m\rangle}V^C_{k,r}}2
  \right)
  \ket{\psi_{k,m}}
  \right].
  \label{eq:overview-common-prediction}
\end{equation}
Crucially, every term in this average involves a matching pair $(r,r)$. By contrast, if the parties simply run their two local GL circuits, the desired branch where both parties are correct corresponds instead to Equation~\eqref{eq:overview-independent-branch}, which averages over
every pair $(r,s)$. Thus, $p_{\mathrm{Pred}}$ gives no direct lower bound on the averaged overlap $\E_{(k,m)\leftarrow\pi}
  \braket{b_{k,m}}{c_{k,m}}$ appearing on the LHS of \eqref{eq:overview-independent-correlation}. 
Indeed, Appendix~\ref{app:cancellation} gives an example in which the averaged
overlap is zero even though $p_{\mathrm{Pred}} > 1/2$. 

A perhaps natural attempt to fix the reduction is to put the two coherent mask registers in the entangled state
\[
  \frac1{2^{n/2}}\sum_r\ket r^R\ket r^S.
\]
This forces Bob's and Charlie's GL extractor calls to use the same internal mask.
The problem is that now the final Hadamard transforms can only ``see'' the \emph{difference} between the two output strings: indeed, for an output pair
$(y,z)$, the two Fourier phases multiply to
\[
  (-1)^{\langle r,y\rangle}(-1)^{\langle r,z\rangle}
  =
  (-1)^{\langle r,y\oplus z\rangle}.
\]
Thus, for the desired pair $(y,z)=(m,m)$, the phase is $1$ for every $r$, but the same is true for every equal pair $(y,y)$, so no Fourier phase remains that distinguishes $(m,m)$ from the other equal pairs.  

\paragraph{The reduction from $\mathsf{GL}(G)$ to $G$.}
We now present the new reduction, taking a strategy for $\mathsf{GL}(G)$ to a strategy for $G$. We will then spend the rest of the technical overview explaining why the reduction works. Fix a known parameter
$0<\gamma\leq\Delta$ (i.e.\ some known lower bound on $\Delta$), and set
$\lambda_\gamma:=\frac1{1+2\gamma}$.
The reduction is as follows.

\begin{minorprotocol}
[Common-mask GL reduction]
\label{alg:gl-reduction-tech}

\smallskip
\begin{enumerate}[label=\arabic*.,leftmargin=2em,itemsep=0.6em]
  \item Keep the splitting procedure $\mathcal A$ unchanged.  
  \item Bob samples a number $L$ according to:
  \[
    \Pr[L=\ell]
    =
    (1-\lambda_\gamma)\lambda_\gamma^\ell,
    \qquad \ell=0,1,2,\ldots \,.
  \]
  Then, he samples a list of $L$ independent and uniformly random masks $T=(t_1,\ldots,t_L)$. 
  
  After learning $k$, Bob applies
  \[
    W^B_{k,T}
    :=
    \begin{cases}
      \Id_B, & L=0,\\
      V^B_{k,t_L}\cdots V^B_{k,t_1}, & L\geq 1.
    \end{cases}
  \]  He then applies the single-party Goldreich-Levin extractor $\GL^B_k$ using a fresh register
  $R$ initialized to $\ket{0^n}$, measures $R$, and outputs the resulting
  string $m_B$.
  \item Charlie just applies $\GL^C_k$ to a fresh register $S$ initialized
  to $\ket{0^n}$, measures $S$, and outputs the resulting string $m_C$.  He
  does not use $T$ or apply any other operation.
\end{enumerate}
\end{minorprotocol}

\paragraph{The main idea: a polynomial ``filter'' to align Bob and Charlie's individual correct-output branches.} \emph{What is this  mysterious reduction achieving?} The only difference from the reduction in the independent-mask setting is that Bob now applies the random sequence of unitaries $$W^B_{k,T}=V^B_{k,t_L}\cdots V^B_{k,t_1}$$ 
before running his Goldreich-Levin extractor. If we follow an analogous analysis as in the independent-mask setting for the new reduction, we see that the probability of successful extraction becomes
$$p_{\mathrm{search}} = \mathbb{E}_T\left\|
    (\widehat V^B_k(m) W^B_{k,T})\otimes \widehat V^C_k(m)\ket{\psi_{k,m}}
  \right\|^2\,.$$
(the only change is the insertion of $W^B_{k,T}$ and the averaging over $T$). Now, there is one degree of freedom that turns out to be crucial in our analysis. Observe that we can rewrite the latter as
$$\mathbb{E}_T\left\|
    (\widehat V^B_k(m) W^B_{k,T})\otimes (U_T^C \widehat V^C_k(m))\ket{\psi_{k,m}}
  \right\|^2\,,$$
for any arbitrary family of unitaries $U_T^C$ acting only on register $C$ (since these do not affect the norm). From here, as in the previous analysis, we can lower-bound the expression, again via Cauchy-Schwarz, by
$$\mathbb{E}_T \left|\braket{b_{k,m}}{W^B_{k,T} \otimes U^C_T \,|c_{k,m}}\right|^2\,.$$
where $\ket{b_{k,m}}$ and $\ket{c_{k,m}}$ are as before. Comparing this expression to the plain overlap $\left|\braket{b_{k,m}}{c_{k,m}}\right|^2$, the new overlap is \emph{filtered} by $W^B_{k,T} \otimes U^C_T$. The nice thing is that we now get to pick any $U_T^C$ of our choice to help us relate the expression to the success probability in the common-mask game. It turns out, as we will attempt to justify shortly, that one good choice of $U_T^C$ is nothing other than $U_T^C = W^C_{k,T}$! 

We can now further lower bound the last expression by Jensen's inequality, and reduce the problem to analyzing
$$\bra{b_{k,m}}
  \mathbb{E}_T\bigl(W^B_{k,T}\otimes W^C_{k,T}\bigr)
\ket{c_{k,m}} =: \bra{b_{k,m}}
  F_{\gamma, k}
\ket{c_{k,m}}
$$
(and relating the latter to the winning probability in the original decision game). 

Now, for fixed $k, \gamma$, the ``filter'' operator $F_{\gamma, k}$ can actually be understood as a \emph{polynomial}. Define 
\[
  \Jop_{k,t}:=V^B_{k,t}\otimes V^C_{k,t},
  \qquad
  \Jop_k:=\E_t[\Jop_{k,t}].
\]
Then, notice that, since the masks in the sequence $T$ are sampled independently of each other,
$$F_{\gamma,k}
  :=
  (1-\lambda_\gamma)
  \sum_{\ell\geq0}\lambda_\gamma^\ell\Jop_k^\ell \,.$$
where this power series is well-defined since $\| J_k \|_{op} \leq 1$.

Crucially, the ``shape'' of the polynomial is determined by the distribution over $L$ (the length of the list $T$). Then, the name of the game is to try pick the right polynomial so that the resulting filter $F_{\gamma,k}$ does a good job at aligning the vectors $\ket{b_{k,m}}$ and $F_{\gamma, k}
\ket{c_{k,m}}$. The filter will do so by applying the appropriate function to the eigenvectors of $\Jop_k$.  In our case, the power series is precisely the Taylor series of an \emph{inverse} function: 
$$ (1-\lambda_\gamma)
  \sum_{\ell\geq0}\lambda_\gamma^\ell\Jop_k^\ell = (1-\lambda_\gamma)(I-\lambda_\gamma \Jop_k)^{-1}  \,.$$
Now, when one is not worried about efficiency of the reduction, there is no issue with having an infinite series (which simply corresponds to arbitrarily long sequences of unitaries applied by Bob). When efficiency of the reduction is a concern, one has to necessarily truncate the polynomial appropriately (since efficiency is determined by the degree of the polynomial), and the hope is to find a good enough low-degree approximation of the desired filter. This idea is reminiscent of the design of quantum algorithms via the QSVT framework~\cite{gilyen2019quantum}, or, even more so, of constructions of ``approximate ground state projectors'' (AGSP) as polynomials in the Hamiltonian terms, e.g.\ in \cite{arad2012improved, anshu2022area}. Below, we do our best to explain why the filter resulting from the particular choice of a geometric distribution succeeds.

\paragraph{Relating the ``filtered'' overlap to the advantage $\Delta$.}
Note that a good choice of filter should not only make the overlap $\bra{b_{k,m}}
  F_{\gamma, k}
\ket{c_{k,m}}$ as large as possible, but it should also allow us to relate this overlap to the advantage $\Delta$ in $\mathsf{GL}(G)$. Eventually, we show that our filter $F_{\gamma, k}$ satisfies:
$$ \left| \mathbb{E}_{k,m} \bra{b_{k,m}}
  F_{\gamma, k}
\ket{c_{k,m}} \right| \geq \gamma \Delta \,,$$
which, in turn, implies $p_{\mathrm{search}} \geq (\gamma \Delta)^2$, based on the previous argument. Since we can run the reduction with any $0<\gamma \leq \Delta$, this ultimately shows $p_{\mathrm{search}} \geq \Delta^4$. So, our first goal will be to rewrite the overlap in a form that is more amenable to being related to $\Delta$.

Fix $k,m$ for now. For readability, we drop the $k,m$ subscripts for now (we will reintroduce them later!). Denote $\ket b = \ket{b_{k,m}}$ and $\ket c = \ket{c_{k,m}}$ and $F = F_{\gamma, k}$. Define
\[
  \ket g:=\frac{\ket b+\ket c}{2},
  \qquad
  \ket h:=\frac{\ket b-\ket c}{2}.
\]
Then $\ket b=\ket g+\ket h$ and $\ket c=\ket g-\ket h$. 
A first convenient observation is the following. We can rewrite
\[
  \bra{b}
  F
\ket{c}
  =
  \bra g F\ket g-\bra h F\ket h
  +\bra h F\ket g -\bra g F\ket h\,,
\]
and since the last two terms are purely imaginary, we have
\begin{equation*}
  \RePart \bra{b}
  F
\ket{c}
  =
  \bra g F\ket g-\bra h F\ket h.
  \label{eq:overview-average-difference}
\end{equation*}
Since $ |\bra{b}
  F
\ket{c}| \geq  \RePart \bra{b}
  F
\ket{c}$, it will suffice for us to bound the latter expression. Eventually, the right choice of $F$ will have to accomplish the following: increase the relative contribution of $\bra g F\ket g$ compared to $\bra h F\ket h$. Now, do these overlaps have any \emph{operational} meaning? Indeed, they do.

Recall that
\begin{align*}
    \ket b &= \ket{b_{k,m}} = \widehat V^B_k(m)\otimes\Id_C\ket{\psi_{k,m}} =: \Bop \ket{\psi} \,\,\,, \textnormal{ and} \\
    \ket c &= \ket{c_{k,m}} = \Id_B \otimes \widehat V^C_k(m)\ket{\psi_{k,m}} =: \Cop \ket{\psi} 
\end{align*}
Let us introduce the operators
\[
  S_{\mathrm{sum}}:=\frac12(\Bop+\Cop),
  \qquad
  R_{\mathrm{diff}}:=\frac12(\Bop-\Cop).
\]
Then, we have $\ket g=S_{\mathrm{sum}}\ket\psi$ and
$\ket h=R_{\mathrm{diff}}\ket\psi$, and thus
\begin{equation}
   \RePart \bra{b}
  F
\ket{c} = 
  \bra\psi S_{\mathrm{sum}}F S_{\mathrm{sum}}\ket\psi
  -\bra\psi R_{\mathrm{diff}}F R_{\mathrm{diff}}\ket\psi.
  \label{eq:overview-filtered-operator-decomposition}
\end{equation}
Observe that the $R_{\mathrm{diff}}$ operator captures the difference between the post-measurement states conditioned on Bob (resp.\ Charlie) correctly guessing $m$ and Charlie (resp.\ Bob) doing nothing. Concretely, $\| R_{\mathrm{diff}} \ket{\psi}\|$ is precisely this distance. It turns out that $R_{\mathrm{diff}}$ is a good proxy for the \emph{disagreement} between Bob and Charlie, in a formal sense. Recall that we defined
\[
  \Jop_{k,t}:=V^B_{k,t}\otimes V^C_{k,t},
  \qquad
  \Jop_k:=\E_t[\Jop_{k,t}].
\]
Note that $\Jop_{k,t}$ is an observable that precisely measures the ``agreement'' between Bob and Charlie in $\mathsf{GL}(G)$ (on key $k$ and mask $t$). This is because the observable has eigenvalue $+1$ on the eigenspace where Bob and Charlie return the same answer, and $-1$ where they return different answers. We can then define the operator $D_k =\frac12 (I-\Jop_k)$, which captures \emph{disagreement}. It turns out that one can show:
\begin{equation*}
  -D_k\preceq R_{\mathrm{diff}}\preceq D_k \,,
  \label{eq:overview-disagreement-controls-difference}
\end{equation*}
i.e.\ $R_{\mathrm{diff}}$ is dominated by the disagreement operator $D_k$.
On the other hand, $S_{\mathrm{sum}}$ does not quite capture ``agreement'', but it simply captures the sum of Bob and Charlie's individual correct-output probabilities. So, at a high level, the filter $F$ should accomplish the following: suppress the negative contribution to the overlap coming from the disagreement, while maintaining enough positive contribution from Bob and Charlie's individual success probabilities.

We now reintroduce the $k,m$ dependence for clarity and write $S^{k,m}_{\mathrm{sum}}$ and $R^{k,m}_{\mathrm{diff}}$. We also introduce the following compact notation:
\[
  \langle X\rangle_\pi
  :=
  \E_{(k,m)\leftarrow\pi}
  \bra{\psi_{k,m}}
  X
  \ket{\psi_{k,m}} \,.
\]
For example, 
$$
  \langle S_{\mathrm{sum}}\rangle_\pi
  :=
  \E_{(k,m)\leftarrow\pi}
  \bra{\psi_{k,m}}
  S^{k,m}_{\mathrm{sum}}
  \ket{\psi_{k,m}} \,.
$$

Now, a crucial observation is that, together, $S_{\mathrm{sum}}$ and $D_k$ are sufficient to determine the advantage $\Delta$. This observation is what lets us connect $\RePart \bra{b}
  F
\ket{c}$ to the advantage $\Delta$. Rewriting
Equation~\eqref{eq:overview-common-prediction} as
\begin{align*}
  p_{\mathrm{Pred}}
  &= 
  \E_{\substack{(k,m)\leftarrow\pi\\
                 r\gets\{0,1\}^n}}\!\left[
  \bra{\psi_{k,m}}
  \left(
    \frac{\Id_B+(-1)^{\langle r,m\rangle}V^B_{k,r}}2
    \otimes
    \frac{\Id_C+(-1)^{\langle r,m\rangle}V^C_{k,r}}2
  \right)
  \ket{\psi_{k,m}}
  \right]\\
  &= \frac 14
  \E_{(k,m)\leftarrow\pi}\!\left[
  \bra{\psi_{k,m}}
  \left(
    \Id_{BC}
    +
    \hat V^B_{k}(m)\otimes I_C
    +
    I_B\otimes \hat V^C_{k}(m) 
    + 
    \E_{\substack{r\gets \{0, 1\}^n}}\!
    V^B_{k,r}\otimes V^C_{k, r}
  \right)
  \ket{\psi_{k,m}}
  \right]\\
  &= \frac 14
  \big\langle 
    \Id_{BC}
    + 
    \Bop_{k,m} + \Cop_{k,m}
    +
    \Jop_k
  \big\rangle_\pi
\end{align*}
and using
$S^{k,m}_{\mathrm{sum}}=\frac12(\Bop_{k,m}+\Cop_{k,m})$, $D_k=\frac12(I-\Jop_k)$ and $\Delta = 2p_{\mathrm{Pred}} - 1$,  we have
\begin{equation}
  \Delta
  =
  \langle S_{\mathrm{sum}}\rangle_\pi -  \langle D\rangle_\pi.
  \label{eq:overview-advantage-sum-disagreement}
\end{equation}
While the above is a straightforward calculation, it establishes a somewhat surprising fact that is worth highlighting: it relates Bob and Charlie's advantage $\Delta$ in the decision game $\mathsf{GL}(G)$ to: their ``disagreement'' in  $\mathsf{GL}(G)$, captured by $D$, and the sum of their individual \emph{correct-extraction} probabilities when running a local Goldreich-Levin extractor, captured by $S_{\mathrm{sum}}$.

Notice now that Equation~\eqref{eq:overview-advantage-sum-disagreement} is quite similar to \eqref{eq:overview-filtered-operator-decomposition}, when we rewrite the latter using the new notation we introduced (and average over $k,m$):
\begin{equation}
\label{eq:12}
\RePart \langle \Bop F \Cop \rangle_{\pi} = \langle S_{\mathrm{sum}}F S_{\mathrm{sum}} \rangle_{\pi} -  \langle R_{\mathrm{diff}}F R_{\mathrm{diff}} \rangle_{\pi}\,.
\end{equation}
This suggests a proof strategy: choose a filter $F$ such that
\begin{enumerate}
\item $\langle R_{\mathrm{diff}}F R_{\mathrm{diff}} \rangle_{\pi} \leq \alpha  \langle D\rangle_{\pi}$, and
\item $\langle S_{\mathrm{sum}}F S_{\mathrm{sum}} \rangle_{\pi} \geq \alpha  \langle S_{\mathrm{sum}}\rangle_{\pi}$
\end{enumerate}
for some factor $\alpha>0$ that is not too small. We will show that one can obtain $\alpha = \gamma$, which implies the desired bound $\RePart \langle \Bop F \Cop \rangle_{\pi} \geq \gamma \Delta$.
\paragraph{Why the chosen filter works.} Recall that we defined $\lambda_\gamma:=\frac1{1+2\gamma}$, and that our reduction uses the particular filter 
\begin{align}
  F_{\gamma,k}
  &=
  (1-\lambda_\gamma)
    \sum_{\ell\geq0}(\lambda_\gamma\Jop_k)^\ell \nonumber\\
 &= (1-\lambda_\gamma)(I-\lambda_\gamma \Jop_k)^{-1} \nonumber\\
  &= f(D_k) \,, \label{eq:29}
\end{align}
where a straightforward calculation gives $f(d) = \frac{\gamma}{d + \gamma}$. 
In other words, $F_{\gamma,k}$ is diagonal in an eigenbasis of $D_k$, and its diagonal elements are obtained by applying $f$ to the eigenvalues of $D_k$. 
Note what such a filter $F$ achieves: it is close to the identity on components with little disagreement ($d<\gamma$) but suppresses, inversely in $d$, components with large disagreement ($d >> \gamma$). In this sense, $F$ is \emph{approximately a projection onto the low-disagreement} component of the state.

\medskip
\noindent \emph{The $R_{\mathrm{diff}}F_{\gamma,k}R_{\mathrm{diff}}$
bound.} We aim to show $$\langle R_{\mathrm{diff}}F R_{\mathrm{diff}} \rangle_{\pi} \leq  \gamma  \langle D\rangle_{\pi}\,.$$ 

Recall that $R_{\mathrm{diff}}$ is dominated by the disagreement operator $D_k$:
\[
  -D_k
  \preceq
  R_{\mathrm{diff}}
  \preceq
  D_k,
\]
Assume for simplicity that $R_{\mathrm{diff}}$ and $D_k$ commute, and thus share an eigenbasis. Then, we have  
\[
  R_{\mathrm{diff}}
  F_{\gamma,k}
  R_{\mathrm{diff}}
  \preceq D_k
  F_{\gamma,k}
  D_k \,.
\]
(since $F$ is positive and also diagonal in any eigenbasis of $D$).
From here, one can immediately see that
\[
  D_k
  F_{\gamma,k}
  D_k  \preceq \gamma D_k \,,
\]
since any component with eigenvalue $d$ gets transformed to $\frac{\gamma d^2}{d+\gamma}$ by the LHS, which is less than $\gamma d$. Thus, this gets us the desired bound. A short block-matrix argument (which we defer to the main body) makes this intuition valid even
when $R_{\mathrm{diff}}$ does not commute with $D_k$.

\medskip
\noindent \emph{The $S_{\mathrm{sum}}F S_{\mathrm{sum}}$
bound.} We are now left with showing 
\[
  \langle S_{\mathrm{sum}}F S_{\mathrm{sum}}\rangle_\pi
  \geq
  \gamma\langle S_{\mathrm{sum}}\rangle_\pi \,.
\]

Suppressing $k,m$ for a moment, let $\ket{x} =S_{\mathrm{sum}}\ket{\psi}$. The difficulty is that the LHS overlap $\bra{x}F\ket{x}$ may be very small due to the filter, even though the RHS overlap $\langle\psi| x \rangle = \langle \psi | S_{\mathrm{sum}}\ket{\psi}$ is large. This can happen if $\ket{x}$ has a lot of its weight on eigenspaces of $F$ corresponding to small eigenvalues.
More concretely, if $F$ has a tiny eigenvalue $\mu$ on some component, then that component
can contribute to $\bra{x}F\ket{x}$ at a ``discount rate'' of $\mu$. More generally, among all vectors $\ket{x}$ having a fixed overlap with $\ket{\psi}$, the vector with smallest contribution to the filtered overlap points in the direction
$F_{\gamma,k}^{-1}\ket{\psi}$. Formally, this follows from a Cauchy-Schwarz inequality:
\[
  \bigl\langle
    S_{\mathrm{sum}}F_{\gamma,k}S_{\mathrm{sum}}
  \bigr\rangle_\pi
  \geq
  \frac{
    \bigl|\langle S_{\mathrm{sum}}\rangle_\pi\bigr|^2
  }{
    \langle F_{\gamma,k}^{-1}\rangle_\pi
  }.
\]

Now, by Equation~\eqref{eq:overview-advantage-sum-disagreement}, the numerator on the RHS is $(\Delta + d)^2$, where $d=\langle D_k\rangle_\pi$. On the other hand, by Equation \eqref{eq:29}, the filter is designed so that the cost of applying the inverse filter is exactly linear in the disagreement $d$ (up to a rescaling by $\gamma$):
\[
  \langle F_{\gamma,k}^{-1}\rangle_\pi
  =
  \frac{\gamma + \langle D_k\rangle_\pi}{\gamma}
  =
  \frac{\gamma+d}{\gamma}.
\]
Using $\gamma \leq \Delta$, we get precisely
\[
  \bigl\langle
    S_{\mathrm{sum}}F_{\gamma,k}S_{\mathrm{sum}}
  \bigr\rangle_\pi
  \geq
  \gamma(\Delta+d)
  =
  \gamma\langle S_{\mathrm{sum}}\rangle_\pi \,.
\]
as desired. 

Finally, putting our $R_{\mathrm{diff}}F_{\gamma,k}R_{\mathrm{diff}}$ and $S_{\mathrm{sum}}F S_{\mathrm{sum}}$ bounds together with \eqref{eq:overview-advantage-sum-disagreement} and \eqref{eq:12} yields the desired bound $\RePart \langle \Bop F \Cop \rangle_{\pi}  \geq \gamma \Delta$. This, in turn gives $p_{\mathrm{search}} \geq \Delta^4$ (choosing $\gamma = \Delta$ in the reduction). 

As a final remark, notice that the filter achieves a very delicate balance: it suppresses high-disagreement components enough to control the negative contribution from $R_{\mathrm{diff}}F_{\gamma,k}R_{\mathrm{diff}}$, while its inverse grows only linearly with disagreement, which is just what is needed to have enough of the positive $S_{\mathrm{sum}}F S_{\mathrm{sum}}$ contribution survive.

\paragraph{Relation to the filter used by Ananth and Sahai~\cite{AS26}.} 
The filter used in our reduction is analogous to a filter used by Ananth and Sahai~\cite{AS26}. Both filters have the geometric series form
\[
  J=(1-\lambda_\gamma)\sum_{\ell=0}^{\infty}\lambda_\gamma^\ell\Jop^\ell = (1-\lambda_\gamma)(\Id-\lambda_\gamma\Jop)^{-1} \,,
\]
but they are used in different ways in their respective analyses. Ananth and Sahai prove that, for every integer $\ell \geq 0$,
\[
 \bigl\lVert\mathsf{E}_B\Jop^\ell\mathsf{E}_C\bigr\rVert_{\mathrm{op}}
\leq
 \bigl\lVert\mathsf{E}_B\mathsf{E}_C\bigr\rVert_{\mathrm{op}}\,,
\]
where $\mathsf{E}_B$ and $\mathsf{E}_B$ are operators capturing Bob and Charlie's individual success probabilities. They then use this inequality to bound certain ``agreement moments'' that arise in their analysis.

For us, the filter has a clear
operational meaning. It arises from Bob's actual sequence of unitaries $W^B_{k,T} = V_{k,t_L} \cdots V_{k,t_1}$ in our Goldreich-Levin reduction (Reduction~\ref{alg:gl-reduction-tech}) (after inserting matching operations on Charlie's side for the purpose of the analysis).

\section{The model and main result}
\label{sec:model}
We start by formulating ``search'' security for an unclonable encryption scheme. This can be modeled generically as a ``cloning game'' as in \cite{AKL23}. 

\begin{definition}[``Search'' unclonable encryption game]
\label{def:search-game}
Fix an integer $n\geq 1$ (which can be understood as the security parameter), a finite key space $\cK$, and a
distribution $\pi$ on $\cK\times\{0,1\}^n$.  
\begin{itemize}
\item (sampling) The challenger samples $(k,m)\leftarrow\pi$.
\item (token state generation) For every $(k,m)$, a challenger prepares a token state
$\rho^A_{k,m}$ on a finite-dimensional register $A$.  
\item (splitting) 
The adversary applies a fixed cloning/splitting procedure $\mathcal A$ producing
\begin{equation*}
  \sigma_{k,m}^{BC}:=\mathcal A(\rho^A_{k,m}).
  \label{eq:split-state}
\end{equation*}
The state $\sigma_{k,m}^{BC}$ may in general be mixed and arbitrarily entangled across $B$ and $C$. We think of registers $B$ and $C$ as going to two parties Bob and Charlie respectively. In the rest of the paper, we will take the state to be pure, and denote it as $\ket{\psi_{k,m}}^{BC}$. This is without loss of generality, since $\mathcal{A}$ can always provide any purifying register to one of the two parties (and this can only increase their success probability).
\item (Key is revealed and Bob and Charlie guess) The challenger reveals the key $k$ to Bob and Charlie, who respectively apply a local measurement on $B$ and $C$ to obtain a guess for $m$ (this measurement can be taken to be projective without loss of generality since $\mathcal{A}$ can provide any required auxiliary register).
\end{itemize}
The game is won if \emph{both} Bob and Charlie's guesses equal $m$.
\end{definition}

\begin{remark}
The simplest example of such a game is based on BB84 states. There, $k$ encodes a basis $\{0,1\}^n$ and $\rho^A_{k, m}$ is the state encoding $m$ under the basis $k$: $\bigotimes_i H^{k_i} \ket{m_i}$. The resulting game, considered by Broadbent and Lord in~\cite{BL20} in the setting of unclonable encryption, is ``secure'' (in the sense of having an exponentially small optimal winning probability). The latter follows from a monogamy-of-entanglement property of BB84 states established in~\cite{TFKW13}.
\end{remark}

We now define a Goldreich-Levin mapping from a search game $G$ as above to a decision game $\mathsf{GL}(G)$. The latter corresponds to the following mapping of an unclonable encryption scheme that satisfies ``search'' security into one that satisfies ``indistinguishability'' security. To encrypt the bit $b$: use the search secure scheme to encrypt a uniformly random string $m \in \{0,1\}^n$ via the ciphertext $|\mathsf{Enc}(k,m)\rangle$; sample another uniformly random ``mask'' $r \in \{0,1\}^n$; the new ciphertext is $\left(|\mathsf{Enc}(k,m)\rangle, r, \langle r, m \rangle \oplus b \right)$ (i.e.\ use $\langle r, m \rangle$ as a one-time pad).  Although this transformation is typical in classical (and quantum) cryptography, its validity for unclonable encryption games has remained a conjecture since~\cite{AKLLZ22}. As discussed, the primary source of difficulty and barrier is that, in the corresponding decision game, the ``mask'' $r$ is identical for Bob and Charlie (since it is part of the original ciphertext), rather than Bob and Charlie having their own independently sampled masks.
\begin{definition}[``Common-mask'' Goldreich-Levin unclonable encryption game]
\label{def:common-mask-game}
Let $G$ be a search unclonable-encryption game as in Definition \ref{def:search-game}. We define the game $\mathsf{GL}(G)$ as follows:
\begin{itemize}
\item The sampling, token state generation, and splitting phases are identical as in $G$. 
\item The only difference is in the final phase: the challenger additionally samples a uniformly random $r \in \{0,1\}^n$, and reveals $k$ and $r$ to Bob and Charlie. Then, Bob and Charlie each return a guess for $\langle r, m \rangle$. The game is won if both guesses are correct. 
\end{itemize}
We refer to $r$ as the ``mask''.
\end{definition}

The main result of this work is that the ``Common-mask'' Goldreich--Levin search-to-decision reduction works for any unclonable encryption game.
\begin{theorem}
\label{thm:main}
The following holds for any search unclonable encryption game $G$ (as in Definition~\ref{def:search-game}). 
Let $p_{\mathrm{Search}}$ and $p_{\mathrm{Pred}}$ be the optimal success probabilities in $G$ and $\mathsf{GL}(G)$ respectively (where $\mathsf{GL}(G)$ is as in Definition~\ref{def:common-mask-game}). Then, 
\begin{equation*}
  p_{\mathrm{Search}}
  \geq
  \bigl(2p_{\mathrm{Pred}}-1\bigr)^4.
  \label{eq:main-bound}
\end{equation*}
Note that $2p_{\mathrm{Pred}}-1$ is (twice) the optimal ``advantage'' over $\frac12$ (i.e.\ random guessing) in $\mathsf{GL}(G)$.

Moreover, there is an explicit reduction that takes any strategy for $\mathsf{GL}(G)$ that wins with probability $\frac12 + \frac 12 \Delta$ (for $\Delta \geq 0$) to a strategy for $G$ that wins with probability $\gamma^2 \Delta^2$, where $\gamma$ is any known lower bound on $\Delta$. The reduction is uniform; and takes expected polynomial time in the runtime of the strategy for $\mathsf{GL}(G)$ and $1/\gamma$.
\end{theorem}
The runtime of the reduction can be made \emph{worst-case} polynomial-time at the cost of a factor of $\frac12$ loss in the $p_{\mathrm{Search}}$ lower bound (see Theorem~\ref{thm:efficient-capped}). Moreover, there is a reduction that does not need to know a lower bound $\gamma$ at all, and achieves $p_{\mathrm{Search}} \geq \Omega(\Delta^6)$ (see Theorem~\ref{thm:unknown-delta}).

Going forward, for convenience, we will often refer to $\Delta$ as the advantage of a strategy, even though it is technically twice the advantage. 

The rest of the paper is devoted to the proof of Theorem~\ref{thm:main}.

\section{A warmup: the independent-mask Goldreich--Levin reduction}
\label{sec:independent}
In this section, we consider a simpler version of the Goldreich-Levin unclonable encryption game from Definition~\ref{def:common-mask-game}. The only difference will be that here Bob and Charlie receive \emph{independently} sampled masks $r$ and $r'$ (rather than identical) and will have to respectively guess $\langle r, m \rangle$ and $\langle r', m \rangle$.

This ``independent-mask'' version of the game has been analyzed before. The analysis follows straightfowardly from the standard quantum Goldreich-Levin reduction of Adcock and Cleve~\cite{AC02}, as observed in, e.g.~\cite{CLLZ21, AKL23}. We recall this analysis here in broad strokes as it will help us introduce some relevant notation, as well as pinpoint the crux in analyzing to the common-mask version of the game.

\subsection{Setup}
All inner products are over $\F_2$, and we write
\begin{equation*}
  \chi_m(r):=(-1)^{\langle r,m\rangle},
  \qquad r\in\{0,1\}^n.
  \label{eq:character}
\end{equation*}
For the rest of this section, we fix the strategy for the ``splitting'' phase. For $k$ and $m \in \{0,1\}^n$ sampled by the challenger, let $\ket{\psi_{k,m}}$ denote the joint state of Bob and Charlie after the splitting phase (which we are assuming is pure without loss of generality). The reduction we will describe does not change the splitting phase at all, only Bob and Charlie's operations.

For key $k$ and mask $r\in\{0,1\}^n$, let $V^B_{k,r}$ and $V^C_{k,r}$ be the binary observables measured by Bob and Charlie respectively to computer their answer bit. Then, their corresponding projective measurements are $\{\Pi^B_{k,r,b}\}_{b \in \{0,1\}}$ and $\{\Pi^C_{k,r,b}\}_{b \in \{0,1\}}$ where
\begin{equation*}
  \Pi^B_{k,r,b}:=\frac{\Id+(-1)^b V^B_{k,r}}2,
  \qquad
  \Pi^C_{k,r,b}:=\frac{\Id+(-1)^b V^C_{k,r}}2.
  \label{eq:decoder-projectors}
\end{equation*}
The ``correct answer'' projectors are
\begin{equation*}
  \Pi^B_{k,r,\langle r,m\rangle}
  =
  \frac{\Id+\chi_m(r)V^B_{k,r}}2,
  \qquad
  \Pi^C_{k,r,\langle r,m\rangle}
  =
  \frac{\Id+\chi_m(r)V^C_{k,r}}2.
  \label{eq:correct-projectors}
\end{equation*}
The winning probability for this strategy is the following.
\begin{definition}[Independent-mask unclonable-encryption game winning probability]
\label{def:indep-game}
\begin{equation*}
  p_{\mathrm{Pred}}^{\mathrm{ind}}
  :=
  \E_{\substack{(k,m)\leftarrow\pi\\
                 r,s\gets\{0,1\}^n}}
  \left[
  \bra{\psi_{k,m}}
  \left(
    \frac{\Id+\chi_m(r)V^B_{k,r}}2
    \otimes
    \frac{\Id+\chi_m(s)V^C_{k,s}}2
  \right)
  \ket{\psi_{k,m}}\right] 
  \label{eq:common-prediction}
\end{equation*}
(where the expectation is over $r,s \in \{0,1\}^n$).
\end{definition}
Note that a winning probability of $1/2$ is trivially achievable: give the ciphertext to Bob and let Charlie do a random guess.  We write
\begin{equation*}
  \Delta^{\mathrm{ind}}:=2p_{\mathrm{Pred}}^{\mathrm{ind}}-1
  \label{eq:Delta}
\end{equation*}
for the advantage over this baseline. We henceforth consider only the case $\Delta^{\mathrm{ind}}\geq0$.

\subsection{Independent-mask GL reduction}
We now describe a reduction that takes a strategy achieving advantage $\Delta^{\mathrm{ind}}$ in the independent-mask unclonable-encryption game, and obtains a strategy that wins with probability at least $(\Delta^{\mathrm{ind}})^2$ in the ``search'' unclonable-encryption game of Definition~\ref{def:search-game}, i.e.\ Bob  and Charlie simultaneously extract $m$ with probability at least $(\Delta^{\mathrm{ind}})^2$. The reduction is simple: Bob and Charlie each locally run the \cite{AC02} Goldreich-Levin extractor, and return the output. We recall how this extractor works.

\paragraph{The \cite{AC02} Goldreich-Levin extractor.}
Let $R$ be an $n$-qubit register. For a key $k$, define Bob's controlled reflection
\begin{equation*}
  \CV^B_k
  :=
  \sum_{r\in\{0,1\}^n}
  \proj r^R\otimes V^B_{k,r} 
  \,.
  \label{eq:controlled-V-B}
\end{equation*}
Then, the \cite{AC02} Goldreich--Levin extractor applies the following unitary
\begin{equation}
  \GL^B_k
  :=
  (H_R^{\otimes n}\otimes\Id_B)
  \CV^B_k
  (H_R^{\otimes n}\otimes\Id_B)\,,
  \label{eq:GL-B}
\end{equation}
followed by a measurement of the $R$ register. The guarantee proven in \cite{AC02}, which follows from a straightforward calculation, is that, for any state $\ket{\phi}^B$ which provides advantage $\epsilon$ at guessing $\langle r, m \rangle$, the extractor above will output $m$ with probability at least $\epsilon^2$ when run on $\ket{0}^R \ket{\phi}^B$. In the independent-mask GL reduction, Bob and Charlie both locally run the \cite{AC02} extractor in parallel. So, let us define $\CV^C_k$ and $\GL^C_k$ analogously.

Now, for $y\in\{0,1\}^n$, define
\begin{equation}
  \widehat V^B_k(y)
  :=
  2^{-n}\sum_r(-1)^{\langle r,y\rangle}V^B_{k,r},
  \qquad
  \widehat V^C_k(y)
  :=
  2^{-n}\sum_r(-1)^{\langle r,y\rangle}V^C_{k,r}.
  \label{eq:Vhat}
\end{equation}
The following lemma will be important later.
\begin{lemma}
\label{lem:GL-expansion} For every state $\ket\phi^B$,
\begin{equation*}
  \GL^B_k
  \bigl(\ket{0^n}^R\ket\phi^B\bigr)
  =
  \sum_{y\in\{0,1\}^n}
  \ket y^R\widehat V^B_k(y)\ket\phi^B.
  \label{eq:GL-expansion}
\end{equation*}
Equivalently, the following operator identity holds for all $y$:
$$
  \widehat V^B_k(y). = (\bra y^R\otimes\Id_B)\GL^B_k
  (\ket{0^n}^R\otimes\Id_B)  \,.
$$
The analogous identities hold for Charlie.
\end{lemma}

\begin{proof}
Fix an arbitrary input state $\ket\phi^B$. Then,
\begin{align*}
  \GL^B_k\bigl(\ket{0^n}^R\ket\phi^B\bigr)
  &= (H_R^{\otimes n}\otimes\Id_B)
  \CV^B_k
  (H_R^{\otimes n}\otimes\Id_B) \\
 &= 2^{-n}\sum_{r,y\in\{0,1\}^n}
  (-1)^{\langle r,y\rangle}
  \ket y^R V^B_{k,r}\ket\phi^B \\
  &=\sum_{y\in\{0,1\}^n}
  \ket y^R\widehat V^B_k(y)\ket\phi^B
\end{align*}
The sums are finite, so we may group together all terms with the same value
of $y$:
\begin{align*}
  \GL^B_k\bigl(\ket{0^n}^R\ket\phi^B\bigr)
  &=
  \sum_{y\in\{0,1\}^n}
  \ket y^R
  \left(
    2^{-n}\sum_{r\in\{0,1\}^n}
    (-1)^{\langle r,y\rangle}V^B_{k,r}
  \right)
  \ket\phi^B \\
  &=
  \sum_{y\in\{0,1\}^n}
  \ket y^R\widehat V^B_k(y)\ket\phi^B \,.
\end{align*}

Finally, projecting the $R$ register onto $\bra{y}$, and noting that the resulting identity holds for all $\ket{\phi}$, yields the operator identity 
$$
  (\bra y^R\otimes\Id_B)\GL^B_k
  (\ket{0^n}^R\otimes\Id_B)
  =
  \widehat V^B_k(y) \,.
$$
The argument is analogous on Charlie's side.
\end{proof}

\paragraph{The independent-mask GL reduction.}
As mentioned earlier, the independent-mask Goldreich-Levin reduction is simple: Bob and Charlie each locally run the \cite{AC02} Goldreich-Levin extractor.

\begin{proposition}[Independent-mask GL reduction]
\label{prop:independent-mask-reduction}
Let $p_{\mathrm{Pred}}^{\mathrm{ind}}$ be the probability that Bob and Charlie win the independent-mask game (as in Definition~\ref{def:indep-game}). Assume $p_{\mathrm{Pred}}^{\mathrm{ind}}\geq1/2$ and set $\Delta^{\mathrm{ind}}
  :=
  2p_{\mathrm{Pred}}^{\mathrm{ind}}-1$.
If Bob and Charlie each run the Goldreich-Levin extractor locally, they both output the correct $m$ with probability at least
\begin{equation*}
  p_{\mathrm{Search}}
  \geq
  (\Delta^{\mathrm{ind}})^2
  =
  \bigl(2p_{\mathrm{Pred}}^{\mathrm{ind}}-1\bigr)^2.
  \label{eq:independent-search-bound}
\end{equation*}
\end{proposition}

\begin{proof}
A direct consequence of Lemma~\ref{lem:GL-expansion} is that, upon running their respective Goldreich-Levin extractors and obtaining the same correct outcome $m$, the unnormalized leftover state on registers $B$ and $C$ is
$$\widehat V^B_k(m)\otimes \widehat V^C_k(m)\ket{\psi_{k,m}}\,.$$
Consequently, the probability that Bob and Charlie output the same correct outcome $m$ given key $k$ in the first place is precisely
\begin{align}
  p_{k,m}
  &:={}
  \left\|
    \widehat V^B_k(m)\otimes \widehat V^C_k(m)\ket{\psi_{k,m}}
  \right\|^2
  \nonumber\\
  &\geq
  \left|
    \bra{\psi_{k,m}}
    \widehat V^B_k(m)\otimes \widehat V^C_k(m)
    \ket{\psi_{k,m}}
  \right|^2
  \nonumber\\
  &=
  \left|\braket{b_{k,m}}{c_{k,m}}\right|^2.
  \label{eq:independent-branch-bound}
\end{align}
where the first inequality is by Cauchy--Schwarz, and in the last equality we defined 
\begin{equation*}
  \ket{b_{k,m}}:= \widehat V^B_k(m) \otimes I\ket{\psi_{k,m}},
  \qquad
  \ket{c_{k,m}}:= I \otimes \widehat V^C_k(m)\ket{\psi_{k,m}}.
  \label{eq:b-c-vectors}
\end{equation*}

Let $X$ and $Y$ denote the $\pm1$-valued random variables for Bob and Charlie respectively guessing their own inner product correctly (where $+1$ denotes a correct guess). Then, notice that for fixed $k,m,r,s$, we have
\begin{equation*}
  \E[XY\mid k,m,r,s]
  =
  \chi_m(r)\chi_m(s)
  \bra{\psi_{k,m}}
    V^B_{k,r}\otimes V^C_{k,s}
  \ket{\psi_{k,m}}.
  \label{eq:independent-signed-correlation}
\end{equation*}
Using the definitions of $\widehat V^B_k(m)$ and $\widehat V^C_k(m)$ from \eqref{eq:Vhat}, we have
\begin{align}
 \E_{(k,m)\leftarrow\pi}
  \braket{b_{k,m}}{c_{k,m}} &= {}
  \E_{\substack{(k,m)\leftarrow\pi\\r,s}}
  \left[
    \chi_m(r)\chi_m(s)
    \bra{\psi_{k,m}}
      V^B_{k,r}\otimes V^C_{k,s}
    \ket{\psi_{k,m}}
  \right] \nonumber\\
  &={}\E[XY]
  \label{eq:independent-factorization}
\end{align}
Now, let $P_{u,v}=\Pr[X=u,Y=v]$. By definition,
$P_{+1,+1}=p_{\mathrm{Pred}}^{\mathrm{ind}}$, and therefore
\begin{align}
  \E[XY]
  &=
  P_{+1,+1}+P_{-1,-1}-P_{+1,-1}-P_{-1,+1}
  \nonumber\\
  &=
  2\bigl(P_{+1,+1}+P_{-1,-1}\bigr)-1
  \nonumber\\
  &\geq
  2p_{\mathrm{Pred}}^{\mathrm{ind}}-1
  =
  \Delta^{\mathrm{ind}}.
  \label{eq:independent-agreement-correlation}
\end{align}
Finally, averaging Equation~\eqref{eq:independent-branch-bound} over
$(k,m)\leftarrow\pi$, applying Jensen's inequality and Equations \eqref{eq:independent-factorization} and \eqref{eq:independent-agreement-correlation}, we have
\begin{align*}
  p_{\mathrm{Search}}
  &\geq
  \E_{(k,m)\leftarrow\pi}
  \left|\braket{b_{k,m}}{c_{k,m}}\right|^2
  \\
  &\geq
  \left|
    \E_{(k,m)\leftarrow\pi}
    \braket{b_{k,m}}{c_{k,m}}
  \right|^2
  \geq
  (\Delta^{\mathrm{ind}})^2 \,,
\end{align*}
as desired.
\end{proof}

\paragraph{What fails in the common-mask setting.} We draw the reader's attention to the key assumption that made the above argument go through: the challenges $r$ and $s$ are drawn independently (and $(\Delta^{\mathrm{ind}})^2$ is the probability that they succeed simultaneously on this independent-challenge distribution). This is what allows us to relate the overlap $\braket{b_{k,m}}{c_{k,m}}$ to $\E[XY | k,m]$ in Equation~\eqref{eq:independent-factorization}. In the ``common-mask''
game, which we will discuss next, the hypothesis is that Bob and Charlie have advantage $\Delta$ when receiving \emph{identical} challenges. Thus, the success hypothesis only controls the diagonal terms $r=s$, and
therefore need not bound the same overlap. In this case, what can happen is that $\ket{b_{k,m}}$ and $\ket{c_{k,m}}$ can each have substantial norm, but their overlap is tiny or even zero. In Appendix~\ref{app:cancellation}, we give an explicit two-dimensional example in which the common-mask simultaneous success probability is $25/48$, yet $\ket{b_{k,m}}$ and $\ket{c_{k,m}}$ are orthogonal. To overcome this, we will have to devise a new reduction that works in the common-mask setting.

\section{The common-mask GL reduction}
\label{sec:reduction}
For a search unclonable encryption game $G$, recall the definition of the common-mask unclonable encryption game $\mathsf{GL}(G)$ from Definition~\ref{def:common-mask-game}: the only difference from the independent-mask game discussed in Section~\ref{sec:independent} is that Bob and Charlie now receive identical challenges $r$. In this section, we will describe how to modify the reduction from the previous section so that it works in the common-mask setting. This will allows us to prove our main Theorem~\ref{thm:main}, relating the success probability in $G$ to the success probability in $\mathsf{GL}(G)$.

\subsection{The reduction algorithm}
\label{sec:reduction-algorithm}
Fix a strategy for Bob and Charlie. We use the same notation as in Section~\ref{sec:independent} to describe it: for key $k$ and mask $r\in\{0,1\}^n$, we let $V^B_{k,r}$ and $V^C_{k,r}$ be the binary observables (i.e.\ reflections) that they measure to compute their answer bit. For key $k$ and $m\in\{0,1\}^n$, we denote by $\ket{\psi_{k,m}}$ the state on Bob and Charlie's registers after the splitting phase. As before, we denote by $p_{\mathrm{Pred}}$ the winning probability of this strategy in $\mathsf{GL}(G)$, and let $\Delta=2p_{\mathrm{Pred}}-1$.

The reduction algorithm is ``slightly'' non-black-box. It needs to know a lower bound $0<\gamma\leq\Delta$ on the advantage. The success probability $p_{\mathrm{Search}}$ of the resulting search strategy will then depend on $\gamma$: we will eventually show $p_{\mathrm{Search}} \geq \gamma^2 \Delta^2$. 

Moreover, the reduction algorithm is actually (fundamentally) randomized. As part of the reduction, Bob samples a non-negative integer $L$ according to the following geometric distribution. Let
\begin{equation*}
  \lambda_\gamma:=\frac1{1+2\gamma}.
  \label{eq:lambda-gamma}
\end{equation*}
Then, 
\begin{equation}
  \Pr[L=\ell]
  =
  (1-\lambda_\gamma)\lambda_\gamma^\ell,
  \qquad \ell=0,1,2,\ldots.
  \label{eq:L-distribution}
\end{equation}

Here is the reduction algorithm.

\begin{minorprotocol}
[Common-mask GL reduction]
\label{alg:gl-reduction}
\smallskip
\textbf{Input:} A strategy for $\mathsf{GL}(G)$ (for which we use the notation introduced above), and a parameter $0<\gamma\leq\Delta$.

\smallskip
\textbf{Output:} The following strategy for $G$.

\medskip
The splitting phase is identical as in $\mathsf{GL}(G)$, and results in Bob and Charlie obtaining a state $\ket{\psi_{k,m}}$ on registers $B,C$, for some key $k$ and $m \in \{0,1\}^n$. Then:

\smallskip
\begin{enumerate}[label=\arabic*.,leftmargin=2em,itemsep=0.6em]
  \item 
  Bob samples $L$ from the distribution in \eqref{eq:L-distribution}.
  He then samples
  $T=(t_1,\ldots,t_\ell)$ where the $t_i \leftarrow \{0,1\}^n$ are independent and uniformly random masks. 

  \item After receiving $k$, Bob applies unitary
  \[
    W^B_{k,T}=V^B_{k,t_\ell}\cdots V^B_{k,t_1},
  \]
  to register $B$. Then Bob initializes a register $R$ to $\ket{0^n}$, and apply the \cite{AC02} Golreich-Levin unitary $\GL^B_k$ (as in Equation~\ref{eq:GL-B}) on registers $RB$. Finally, he measures $R$ in the standards basis, and returns the output.

  \item After receiving $k$, Charlie initializes a register $S$ to $\ket{0^n}$. He applies $\GL^C_k$ on registers $SC$, measures $S$ in the computational basis, and returns the output. Note that, unlike Bob, Charlie does not apply any other unitary before $\GL^C_k$.  
\end{enumerate}

\end{minorprotocol}

\medskip
\noindent We will spend most of the remainder of this paper proving the following theorem.

\begin{theorem}
\label{thm:parametric-reduction}
Let $p_{\mathrm{Pred}}$ be the winning probability of a strategy $S$ for game $\mathsf{GL}(G)$. Let $\Delta = 2p_{\mathrm{Pred}}-1$, and assume $\Delta>0$. Let $0<\gamma\leq\Delta$, and let $S'$ be the strategy for game $G$ obtained by running Reduction~\ref{alg:gl-reduction} on input $S$ and $\gamma$. Let $p_{\mathrm{Search}}$ be the winning probability of $S'$ in $G$. Then, 
\begin{equation*}
  p_{\mathrm{Search}}
  \geq
  \gamma^2\Delta^2.
  \label{eq:parametric-success}
\end{equation*}
\end{theorem}

For a high level (but detailed) explanation of what this reduction is accomplishing, see the Technical Overview first (Section~\ref{sec:tech-overview}).

\begin{remark}[GL reduction for monogamy-of-entanglement games]
\label{rem:1}
We remark that this reduction, and its analysis, are essentially unchanged for \emph{monogamy-of-entanglement} games. In a monogamy game, it is $\mathcal{A}$ who first prepares a state on registers $ABC$, sending $A$ to the challenger, and $B$ and $C$ to Bob and Charlie respectively. Then, the challenger measures $A$ in a basis determined by the key $k$, obtaining an outcome $m$. Everything else is identical. One can then analogously define a decision monogamy game $\mathsf{GL}(G)$ from a search monogamy game $G$. The GL reduction for monogamy games is analogous to Reduction~\ref{alg:gl-reduction}: keep $\mathcal{A}$ unchanged, and have Bob and Charlie act exactly as in Reduction~\ref{alg:gl-reduction}. The analysis of this reduction is identical, and achieves the same guarantee, 
since the analysis only needs to consider the state on $BC$ and the respective operations performed by Bob and Charlie on it.
\end{remark}

\section{Analysis of the reduction, part 1: the ``filtered'' overlap}
\label{sec:bridge}
In this section, we begin the analysis of Reduction~\ref{alg:gl-reduction}. We will formalize the high-level argument from the Technical Overview, lower-bounding the probability of successful joint extraction in terms of an appropriately ``filtered'' overlap. We fix a strategy for $\mathsf{GL}(G)$, and use the notation introduced at the start of Section~\ref{sec:reduction} to describe it.

For a key $k$ and a mask $t\in\{0,1\}^n$, define
\begin{equation}
  \Jop_{k,t}
  =
  V^B_{k,t}\otimes V^C_{k,t},
  \qquad
  \Jop_k
  =
  \E_t[\Jop_{k,t}].
  \label{eq:Jt-Gk}
\end{equation}
The operator $\Jop_k$ is Hermitian and satisfies
$\lVert \Jop_k\rVert_{\mathrm{op}}\leq 1$ since $V^B_{k,t}$ and $V^C_{k,t}$ are reflections.  For a finite sequence
$T=(t_1,\ldots,t_\ell)$, define the unitaries
\begin{equation}
  W^B_{k,T}
  =
  V^B_{k,t_\ell}\cdots V^B_{k,t_1},
  \qquad
  W^C_{k,T}
  =
  V^C_{k,t_\ell}\cdots V^C_{k,t_1}.
  \label{eq:sequence-unitaries}
\end{equation}
We let both unitaries be the identity when $T$ is empty. 

As before, we denote by $p_{\mathrm{Pred}}$ the winning probability of the strategy in $\mathsf{GL}(G)$, and let $\Delta=2p_{\mathrm{Pred}}-1$. For a reduction parameter $0<\gamma\leq\Delta$, recall that we defined $ \lambda_\gamma:=\frac1{1+2\gamma}$. Here, we also introduce the following operator, which will be crucial in our analysis: 
\begin{equation*}
  F_{\gamma,k}
  :=
  (1-\lambda_\gamma)
  \sum_{\ell=0}^{\infty}
  \lambda_\gamma^\ell \Jop_k^\ell,
  \label{eq:F-global-preview}
\end{equation*}
where $\Jop_k^0=\Id_{BC}$.  Since $0<\lambda_\gamma<1$ and
$\lVert \Jop_k\rVert\leq1$, the series converges in operator norm.

We have the following straightforward identities.
\begin{lemma}
\label{lem:sequence-power}
For $\ell \in \mathbb{N}$,
\[
  \E_{T:\,|T|=\ell}
  \left[
    W^B_{k,T}\otimes W^C_{k,T}
  \right]
  =
  \Jop_k^\ell.
\]
Consequently,
\[
  \E_T
  \left[
    W^B_{k,T}\otimes W^C_{k,T}
  \right]
  =
  F_{\gamma,k}.
\]
\end{lemma}

\begin{proof}
Let $X_t:=V^B_{k,t}\otimes V^C_{k,t}$.  The definitions of the
synchronized-sequence unitaries in
Equation~\eqref{eq:sequence-unitaries}, including their matching order, give
\[
  W^B_{k,T}\otimes W^C_{k,T}
  =
  X_{t_\ell}\cdots X_{t_1}.
\]
Successively averaging the independent variables
$t_\ell,\ldots,t_1$ yields
\[
  \E_{T: |T| = \ell}[X_{t_\ell}\cdots X_{t_1}]
  =
  \bigl(\E_t X_t\bigr)^\ell
  =
  \Jop_k^\ell.
\]
Averaging this conditional identity over the geometric length law
$\Pr[L=\ell]=(1-\lambda_\gamma)\lambda_\gamma^\ell$ yields
$F_{\gamma,k}$, as desired.
\end{proof}

Now, for fixed $(k,m,T)$, we can apply Lemma~\ref{lem:GL-expansion} to deduce that the unnormalized leftover state on $BC$ conditioned on both Bob and Charlie's outputs being $m$ is 
$$
  \ket{\eta_{k,m,T}}
  =
  \left(
    \widehat V^B_k(m)W^B_{k,T}
    \otimes \widehat V^C_k(m)
  \right)
  \ket{\psi_{k,m}}.
$$
Thus, the correct-output probability is $\|\eta_{k,m,T}\|^2$. 

Now, as informally outlined at the end of Section~\ref{sec:reduction}, a key step before applying Cauchy-Schwarz is to leverage a unitary degree of freedom on Charlie's side. We let 
$$\ket{\eta'_{k,m,T}} =  \left(
    \widehat V^B_k(m)W^B_{k,T}
    \otimes W^C_{k,T} \widehat V^C_k(m)
  \right)
  \ket{\psi_{k,m}} \,.
$$
Then, we have 
\begin{align}
  \|\eta_{k,m,T}\|^2 &= \|\eta'_{k,m,T}\|^2  \nonumber \\
  &\geq
  |\braket{\psi_{k,m}}{\eta'_{k,m,T}}|^2.
  \label{eq:norm-return-bound}
\end{align}

\begin{lemma}
\label{lem:success-overlap}
Let $p_{\mathrm{Search}}$ denote the success probability of the strategy obtained from the
reduction of Reduction~\ref{alg:gl-reduction}.
Then,
\begin{equation*}
  p_{\mathrm{Search}}
  \geq
  \left|
    \E_{(k,m) \leftarrow \pi}
  \bra{\psi_{k,m}}
  \big(\widehat V^B_k(m)\otimes\Id_C \big)
 F_{\gamma,k} \big(
    \Id_B\otimes\widehat V^C_k(m)\big)
  \ket{\psi_{k,m}} 
  \right|^2 \,,
  \label{eq:success-filtered-overlap}
\end{equation*}
\end{lemma}

\begin{proof}
Averaging Equation~\eqref{eq:norm-return-bound} and applying Jensen's inequality
gives
\begin{align}
  p_{\mathrm{Search}}
  &\geq
  \E_{k,m,T}|\braket{\psi_{k,m}}{\eta'_{k,m,T}}|^2
  \nonumber\\
  &\geq
  \left|
    \E_{k,m,T} \braket{\psi_{k,m}}{\eta'_{k,m,T}}
  \right|^2.
  \label{eq:Jensen-amplitude}
\end{align}
By Lemma~\ref{lem:sequence-power}, the average inside the absolute value is
\begin{align*}
  \E_{k,m,T} \braket{\psi_{k,m}}{\eta'_{k,m,T}}
  &= \E_{k,m,T} \bra{\psi_{k,m}} \left(
    \widehat V^B_k(m)W^B_{k,T}
    \otimes W^C_{k,T} \widehat V^C_k(m)
  \right)
  \ket{\psi_{k,m}} \\
  &=\E_{k,m}
  \bra{\psi_{k,m}}
  \big(\widehat V^B_k(m)\otimes\Id_C \big)
 F_{\gamma,k} \big(
    \Id_B\otimes\widehat V^C_k(m)\big)
  \ket{\psi_{k,m}} \,,
\end{align*}
as desired.
\end{proof}

It will be convenient for us to rewrite the lower bound in Equation~\eqref{eq:Jensen-amplitude} more compactly as a filtered ``trace overlap'' (squared) in the form $\left| \Tr(\varrho\Bop F_\gamma\Cop)\right|^2$, for an appropriately chosen  mixed state $\rho$ and operators $\Bop, F_\gamma, \Cop$. Define the following block-diagonal Hermitian operators:
\begin{align}
  \varrho
  &:=
  \bigoplus_{k,m}
  \pi(k,m)\cdot
  \proj{\psi_{k,m}}^{BC}, \nonumber\\
  \Bop
  &:=
  \bigoplus_{k,m}
  \left(
    \widehat V^B_k(m)\otimes\Id_C
  \right), \nonumber\\
  \Cop
  &:=
  \bigoplus_{k,m}
  \left(
    \Id_B\otimes\widehat V^C_k(m)
  \right), \nonumber\\
  \Jop
  &:=
  \bigoplus_{k,m} \, \Jop_k \,,
  \label{eq:block-operators}
\end{align}
where $\Jop_k$ as defined in Equation~\eqref{eq:Jt-Gk}. Note that all of these operators have norm at most $1$. Note also that the $(k,m)$ block of $\Jop$ depends on $k$ but not on $m$. Finally, define also
\begin{equation}
\label{eq:block-operators-F}
  F_\gamma
  =
  (1-\lambda_\gamma)
  \sum_{\ell=0}^{\infty}
  \lambda_\gamma^\ell\Jop^\ell
  =
  \bigoplus_{k,m} F_{\gamma,k} \,.
\end{equation}
Then, we have the following.
\begin{corollary}[Lower bounding search success by a ``filtered'' overlap]
\label{cor:success-overlap-trace}
Let $p_{\mathrm{Search}}$ denote the success probability of the strategy obtained from the
reduction of Reduction~\ref{alg:gl-reduction}.
Then,
\begin{equation*}
  p_{\mathrm{Search}}
  \geq \left|\Tr(\varrho\Bop F_\gamma\Cop)\right|^2 \,.
\end{equation*}
\end{corollary}
\begin{proof}
This follows immediately from Lemma~\ref{lem:success-overlap}, by rewriting the inner product as a trace and a direct computation.
\end{proof}

\section{Analysis of the reduction, part 2: lower bounding the ``filtered overlap'' by the original winning probability}
\label{sec:overlap-bound}
Having lower-bounded the winning probability $p_{\mathrm{Search}}$ of the search strategy by (the square of) the filtered overlap $\Tr(\varrho\Bop F_\gamma\Cop)$ in Corollary~\ref{cor:success-overlap-trace}, we are left with lower-bounding $\Tr(\varrho\Bop F_\gamma\Cop)$ itself. We will show the following.

\begin{lemma}
\label{lem:2}
Let $\rho, \Bop, \Cop, F_{\gamma}$ be as defined in \eqref{eq:block-operators} and \eqref{eq:block-operators-F}. Let $p_{\mathrm{Pred}}$ be the winning probability of the strategy for $\mathsf{GL}(G)$. Let $\Delta=2p_{\mathrm{Pred}}-1$, and assume $\Delta>0$. Let $0<\gamma\leq\Delta$. Then, 
$$ \left|\Tr(\varrho\Bop F_\gamma\Cop)\right| \geq \gamma\Delta \,.$$
\end{lemma}

\noindent Together, Lemma~\ref{lem:2} and Corollary~\ref{cor:success-overlap-trace} prove Theorem~\ref{thm:parametric-reduction} (which completes the proof of Theorem~\ref{thm:main}). So, the rest of this section is devoted to the proof of Lemma~\ref{lem:2}. 

\subsection{Some setup and useful facts}
Recall the definition of $\rho, \Bop, \Cop, \Jop, F_{\gamma}$ from Equations~\eqref{eq:block-operators} and \eqref{eq:block-operators-F}. For $u,v\in\{+1,-1\}$, define
\[
  \Gamma_{u,v}
  :=
  \frac14
  \left(
    \Id+u\Bop+v\Cop+uv\Jop
  \right).
\]
Every $\Gamma_{u,v}$ is positive semidefinite.  Indeed, its $(k,m)$ block is
\[
  \E_t
  \left[
    \frac{\Id+u\cdot \chi_m(t) \cdot V^B_{k,t}}2
    \otimes
    \frac{\Id+v\cdot \chi_m(t) \cdot V^C_{k,t}}2
  \right],
\]
and each tensor factor is a projection since the $V^B_{k,t}$ and $V^C_{k,t}$ are reflections. The signs $u$ and
$v$ record whether Bob and Charlie are correct or incorrect. In particular, note that the winning probability of the strategy in $\mathsf{GL}(G)$ is exactly
\[
  p_{\mathrm{Pred}}
  =
  \Tr(\varrho\Gamma_{+,+}).
\]
Set
\[
  S:=\frac12(\Bop+\Cop),
  \qquad
  D:=\frac12(\Id-\Jop),
  \qquad
  R:=\frac12(\Bop-\Cop).
\]
Letting $\Delta=2p_{\mathrm{Pred}}-1$, and using the expression for
$p_{\mathrm{Pred}}$ above, and $\Tr(\varrho)=1$, we can write the
advantage (over $\frac12$) directly in terms of $S$ and $D$:
\begin{align}
  \Delta
  &=
  2p_{\mathrm{Pred}}-1
  \notag\\
  &=
  \frac12\Tr\!\left[
    \varrho(\Id+\Bop+\Cop+\Jop)
  \right]-1
  \notag\\
  &=
  \frac12\Tr\!\left[
    \varrho(\Bop+\Cop+\Jop-\Id)
  \right]
  \notag\\
  &=
  \Tr\!\left(\varrho(S-D)\right).
  \label{eq:decision-identity}
\end{align}
This identity nicely separates the sum $S$ of the two local averages $\Bop$ and $\Cop$ from
the ``disagreement'' operator $D$.

Since $\Jop\preceq\Id$, we have $D\succeq0$.  Importantly, we also have
\begin{equation}
\label{eq:DplusR}
  D+R=2\Gamma_{+,-}\succeq0,
  \qquad
  D-R=2\Gamma_{-,+}\succeq0.
\end{equation}

\subsection{A key idea: the geometric filter as an inverse}
Recall that
\[
  F_\gamma
  =
  (1-\lambda_\gamma)
  \sum_{\ell=0}^{\infty}\lambda_\gamma^\ell\Jop^\ell,
  \qquad
  \lambda_\gamma=\frac{1}{1+2\gamma},
  \qquad
  D=\frac12(\Id-\Jop).
\]
Since $\|\Jop\|\leq1$ and $0<\lambda_\gamma<1$, we have
$\|\lambda_\gamma\Jop\|<1$, which implies that the geometric series converges in operator norm. In particular, the classic geometric sum identity, extended to sums of commuting operators, gives
\[
  \sum_{\ell=0}^{\infty}
  \lambda_\gamma^\ell\Jop^\ell
  =
  (\Id-\lambda_\gamma\Jop)^{-1}.
\]
Next, substituting $\lambda_\gamma=(1+2\gamma)^{-1}$ gives
\[
  1-\lambda_\gamma
  =
  \frac{2\gamma}{1+2\gamma},
  \qquad
  \Id-\lambda_\gamma\Jop
  =
  \frac{1}{1+2\gamma}
  \bigl((1+2\gamma)\Id-\Jop\bigr).
\]
Combining these identities and then using
$(1+2\gamma)\Id-\Jop=(\Id-\Jop)+2\gamma\Id
=D+2\gamma\Id$, we obtain
\begin{align}
  F_\gamma
  &=
  (1-\lambda_\gamma)
  (\Id-\lambda_\gamma\Jop)^{-1}
  \nonumber\\
  &=
  2\gamma
  \bigl((1+2\gamma)\Id-\Jop\bigr)^{-1}
  \nonumber\\
  &=
  \gamma(D+\gamma\Id)^{-1}.
  \label{eq:F-resolvent}
\end{align}
Because $D\succeq0$, we have $D+\gamma\Id \succ 0$.
Moreover, $D+\gamma\Id\succeq\gamma\Id$, so the last expression also
shows that $0\prec F_\gamma\preceq\Id$. We are now ready to tackle the proof that $\left|\Tr(\varrho\Bop F_\gamma\Cop)\right| \geq \gamma\Delta$. In fact, we will show the stronger statement that $\RePart\Tr(\varrho\Bop F_\gamma\Cop)$ already satisfies the bound.
\begin{lemma}
\label{lem:filtered-overlap}
  $\RePart\Tr(\varrho\Bop F_\gamma\Cop)
  \geq
  \gamma\Delta$.
\end{lemma}

\begin{proof}
First, since
$\Bop=S+R$ and $\Cop=S-R$, we have
\begin{align*}
  \Tr(\varrho\Bop F_\gamma\Cop)
  ={}&
  \Tr(\varrho SF_\gamma S)
  -\Tr(\varrho SF_\gamma R)
  \\
  &+
  \Tr(\varrho RF_\gamma S)
  -\Tr(\varrho RF_\gamma R).
\end{align*}
Since $S$, $R$, and $F_\gamma$ are Hermitian, the first and last terms are
real, while
\[
  \Tr(\varrho RF_\gamma S)
  =
  \overline{\Tr(\varrho SF_\gamma R)}.
\]
Thus the mixed terms are purely imaginary, and taking real parts gives
\begin{equation*}
\label{eq:s-r}
  \RePart\Tr(\varrho\Bop F_\gamma\Cop)
  =
  \Tr\!\left[\varrho SF_\gamma S\right] -\Tr\!\left[\varrho RF_\gamma R\right].
\end{equation*}
We will spend the rest of the proof bounding $\Tr\!\left[\varrho SF_\gamma S\right]$ and $\Tr\!\left[\varrho RF_\gamma R\right]$ separately. 

We first lower-bound $\Tr\!\left[\varrho SF_\gamma S\right]$. Since $F_\gamma\succ0$ (by \eqref{eq:F-resolvent}), Cauchy--Schwarz for the Hilbert-Schmidt inner product and the cyclic property of the trace give 
\begin{align*}
  \bigl|\Tr(\varrho S)\bigr|^2
  &=
  \left|\Tr\left((F_\gamma^{-1/2} \rho^{1/2}) \cdot (F_\gamma^{1/2}\, S \, \rho^{1/2})^\dagger\right)\right|^2
  \\
  &=
  \left|
    \left\langle
      F_\gamma^{-1/2}\varrho^{1/2},
      F_\gamma^{1/2}S\varrho^{1/2}
    \right\rangle_{\mathrm{HS}}
  \right|^2
  \\
  &\leq
  \Tr(\varrho F_\gamma^{-1})
  \Tr(\varrho SF_\gamma S).
\end{align*}
Rearranging gives
\begin{equation}
\label{eq:s}
  \Tr(\varrho SF_\gamma S)
  \geq
  \frac{\bigl|\Tr(\varrho S)\bigr|^2}
       {\Tr(\varrho F_\gamma^{-1})}
\end{equation}
We will simplify both numerator and denominator on the RHS. Let $p_{\neq}$ denote the probability that Bob and Charlie output different bits. Let $X$ and $Y$ denote the $\pm1$-valued random variables for Bob and Charlie respectively guessing their own inner product correctly (where $+1$ denotes a correct guess). So, for fixed
$(k,m,t)$, we have
\begin{align*}
  \E[XY\mid k,m,t]
  &=
  \bra{\psi_{k,m}}
  \left(
    \chi_m(t)V^B_{k,t}\otimes
    \chi_m(t)V^C_{k,t}
  \right)
  \ket{\psi_{k,m}}
  \\
  &=
  \bra{\psi_{k,m}}
  \left(
    V^B_{k,t}\otimes V^C_{k,t}
  \right)
  \ket{\psi_{k,m}},
\end{align*}
Averaging over $(k,m,t)$
and using the definitions of $\varrho$ and $\Jop$ (which can be thought of as an ``agreement'' operator), gives
\[
  \E[XY]=\Tr(\varrho\Jop).
\]
Moreover, $XY=+1$ with probability $1-p_{\neq}$ and $XY=-1$ with probability
$p_{\neq}$, so $\E[XY]=1-2p_{\neq}$.  Recalling that $D=\frac 12(\Id-\Jop)$, we therefore obtain
\begin{equation}
\label{eq:rhoD}
  \Tr(\varrho D)
  =
  \frac12\bigl(1-\Tr(\varrho\Jop)\bigr)
  =
  \frac12\bigl(1-\E[XY]\bigr)
  =
  p_{\neq}.
\end{equation}
Combining Equation~\eqref{eq:decision-identity} with
the above gives
\begin{equation}
\label{eq:trace-rhoS}
  \Tr(\varrho S)=\Delta+p_{\neq} \,. 
\end{equation}
By Equation~\eqref{eq:F-resolvent} and the invertibility of
$D+\gamma\Id$,
\[
  F_\gamma^{-1}
  =
  \frac{1}{\gamma}(D+\gamma\Id)
  =
  \Id+\frac{D}{\gamma}.
\]
Hence, since $\varrho$ is a state and $\Tr(\varrho D)=p_{\neq}$,
\begin{equation}
\label{eq:5}
  \Tr(\varrho F_\gamma^{-1})
  =
  1+\frac{p_{\neq}}{\gamma}.
\end{equation}
Substituting \eqref{eq:trace-rhoS} and \eqref{eq:5} back into \eqref{eq:s} yields
\begin{equation}
\label{eq:sfs-bound}
\Tr(\varrho SF_\gamma S)
  \geq
  \frac{4\gamma(\Delta+p_{\neq})^2}{\gamma+p_{\neq}}.
\end{equation}

We now turn to upper bounding $\Tr\!\left[\varrho RF_\gamma R\right]$. First, note that
\[
  \begin{pmatrix}
    D & R\\
    R & D
  \end{pmatrix}
  \succeq 0 \,,
\]
since the operator is unitarily equivalent to
$(D+R)\oplus(D-R)$, and both $D+R$ and $D-R$ are positive semidefinite by \eqref{eq:DplusR}. Since $\begin{pmatrix}
    0 & 0\\
    0 & \gamma\Id
  \end{pmatrix}
  \succ0$,
we also have
\[
  \begin{pmatrix}
    D & R\\
    R & D+\gamma\Id
  \end{pmatrix}
  \succeq0.
\]
Because $D\succeq0$ and $\gamma>0$, the lower-right block
$D+\gamma\Id$ is strictly positive (and hence invertible). Thus, the Schur-complement criterion applied to
\[
  \begin{pmatrix}
    D & R\\
    R & D+\gamma\Id
  \end{pmatrix}
  \succeq0
\]
gives
\[
  D-R(D+\gamma\Id)^{-1}R\succeq0,
\]
or equivalently $R(D+\gamma\Id)^{-1}R\preceq D$. Multiplying by $\gamma$ and using $F_\gamma=\gamma(D+\gamma\Id)^{-1}$ yields 
\begin{equation*}
\label{eq:rfr}
RF_\gamma R\preceq \gamma D \,.
\end{equation*}

Since $\varrho\succeq0$, the latter implies
\begin{equation}
\label{eq:rfr-bound}
     \Tr(\varrho RF_\gamma R)
  \leq
  \gamma\Tr(\varrho D)
  =
  \gamma p_{\neq} \,,
\end{equation}
where the last equality is by Equation \eqref{eq:rhoD}.

Finally, combining the $S$-term lower bound of \eqref{eq:sfs-bound} with the $R$-term upper bound of \eqref{eq:rfr-bound} yields
\begin{align*}
  \RePart\Tr(\varrho\Bop F_\gamma\Cop)
  &= \Tr\!\left[\varrho SF_\gamma S\right] -\Tr\!\left[\varrho RF_\gamma R\right] \nonumber \\
  &\geq
  \frac{\gamma(\Delta+p_{\neq})^2}{\gamma+p_{\neq}}-\gamma p_{\neq} \,. \label{eq:delta-p}
\end{align*}
Finally, since $0<\gamma\leq\Delta$ and $p_{\neq}\geq0$,
we have
\[
  \frac{\gamma(\Delta+p_{\neq})^2}{\gamma+p_{\neq}}-\gamma p_{\neq} 
  \geq
  \Delta(\gamma+p_{\neq}) - \gamma p_{\neq} \geq \gamma \Delta,
\]
which yields
\[
  \RePart\Tr(\varrho\Bop F_\gamma\Cop)
  \geq
  \gamma\Delta.
\]
\end{proof}
Lemma~\ref{lem:2} implies Lemma~\ref{lem:filtered-overlap}, which, together with Corollary~\ref{cor:success-overlap-trace}, completes the proof of Theorem~\ref{thm:parametric-reduction} (and Theorem~\ref{thm:main}).

\section{Uniformity and efficiency analysis}
\label{sec:completion}
Theorem~\ref{thm:parametric-reduction} establishes the success guarantee for Reduction~\ref{alg:gl-reduction}. In this section, we study its runtime. Theorem~\ref{thm:efficient} below says that, given a lower bound on $0 < \gamma \leq \Delta$, this reduction is black-box, uniform, and has an \emph{expected} runtime polynomial in $\frac{1}{\gamma}$ and in the runtime of the original strategy. Here, ``black-box'' means that, apart from the numerical parameter $\gamma$, the reduction uses the strategy for $\GL(G)$ only via an application of the splitting channel $\mathcal A_n$, and coherent controlled applications of the eflections $V^B_{k,r}$ and $V^C_{k,r}$; it does not use any additional structure of the strategy. Theorem~\ref{thm:efficient-capped} says that an appropriate truncation of Reduction~\ref{alg:gl-reduction} has the desired \emph{worst-case} runtime. 

Finally, we describe a uniform reduction that does \emph{not} require knowing a lower bound on $\Delta$. Theorem~\ref{thm:unknown-delta} says that this reduction can be made to have expected runtime runtime \emph{linear} in that of the original strategy, at the cost of a polynomial loss in the resulting search success probability. 

\begin{theorem}
\label{thm:efficient}
Suppose the strategy for $\GL(G)$ has a uniform circuit implementation, and that an efficiently computable parameter $0<\gamma\leq\Delta$ is given. Then Reduction~\ref{alg:gl-reduction} is uniform and has expected running time polynomial in $1/\gamma$ and in the running time of the strategy for $\GL(G)$.
\end{theorem}

\begin{proof}
We examine Reduction~\ref{alg:gl-reduction} step by step.
\begin{enumerate}[label=\arabic*.,leftmargin=2em,itemsep=0.6em]
  \item The splitting procedure is exactly the same as in the common-mask GL unclonable-encryption game. Thus, the reduction uses the original uniform circuit for $\mathcal A_n$ without modification.

  \item Bob samples $L$ and $T$. Conditioned on $L=\ell$, he samples $\ell$ independent uniform $n$-bit masks and applies
  \[
    W^B_{k,T}=V^B_{k,t_\ell}\cdots V^B_{k,t_1}
  \]
  to register $B$. He then initializes a register $R$ to $\ket{0^n}$, applies the Goldreich--Levin unitary $\GL^B_k$ to registers $RB$, measures $R$ in the computational basis, and returns the outcome. Conditioned on $L=\ell$, this requires $\ell$ applications of Bob's eflections for $W^B_{k,T}$, one further coherent controlled application inside $\GL^B_k$, and time $O(n\ell)$ to sample the masks. Hence, its running time is polynomial in $n$, $\ell$, and the running time of Bob's decoder.

  \item Charlie applies $\GL^C_k$ to his registers and measures the output register in the computational basis. This has a uniform circuit implementation with running time polynomial in $n$ and the running time of Charlie's decoder.
\end{enumerate}

The geometric random variable $L$ can be sampled by repeated independent Bernoulli trials with success probability $1-\lambda_\gamma$. Since $L$ has the distribution in Equation~\eqref{eq:L-distribution} and $\lambda_\gamma=(1+2\gamma)^{-1}$,
\[
  \E[L]
  =
  \frac{\lambda_\gamma}{1-\lambda_\gamma}
  =
  \frac1{2\gamma}.
\]
Thus, the expected number of sampled masks and decoder-reflection applications is $O(1/\gamma)$. The claimed uniform expected-running-time bound follows. In particular, if $\gamma$ is inverse polynomial in $n$ and the original strategy is polynomial time, then the resulting strategy is expected polynomial time in $n$.
\end{proof}

We can obtain a worst-case polynomial-time variant by just modifying Bob's sampling of $L$ (by setting a maximum value). This will achieve the desired worst-case runtime at the cost of a small loss in the search success probability. Define
\begin{equation*}
  L_{\max}
  :=
  \left\lceil
    \frac{\log(2/\gamma^4)}
         {\log(1+2\gamma)}
  \right\rceil.
  \label{eq:Lmax}
\end{equation*}
Bob generates at most $L_{\max}$ independent Bernoulli bits, each equal to $0$ with probability $\lambda_\gamma$ and to $1$ with probability $1-\lambda_\gamma$. If the first $1$ occurs after exactly $\ell<L_{\max}$ preceding zeros, he sets $L=\ell$ and continues; if all $L_{\max}$ bits are zero, he aborts. Let $p_{\mathrm{Search}}^{(<L_{\max})}$ denote the joint winning probability of this truncated procedure.

\begin{theorem}[Truncated variant]
\label{thm:efficient-capped}
Suppose the strategy for $\GL(G)$ has a uniform circuit implementation, and that an efficiently computable parameter $0<\gamma\leq\Delta$ is given. Then, the ``truncated'' procedure described above is uniform, has worst-case running time polynomial in $1/\gamma$ and in the running time of the strategy for $\GL(G)$, and satisfies
\begin{equation*}
  p_{\mathrm{Search}}^{(<L_{\max})}
  \geq
  \frac12\gamma^2\Delta^2
  \geq
  \frac12\gamma^4.
  \label{eq:efficient-success}
\end{equation*}
\end{theorem}

\begin{proof}
The probability that the original geometric sampler produces $L\geq L_{\max}$ is
\[
  (1-\lambda_\gamma)
  \sum_{\ell=L_{\max}}^{\infty}
  \lambda_\gamma^\ell
  =
  \lambda_\gamma^{L_{\max}}.
\]
Aborting on this event can decrease the winning probability by at most its probability. Therefore, Theorem~\ref{thm:parametric-reduction} gives
\begin{align*}
  p_{\mathrm{Search}}^{(<L_{\max})}
  &\geq
  p_{\mathrm{Search}}-\lambda_\gamma^{L_{\max}} \\
  &\geq
  \gamma^2\Delta^2-\lambda_\gamma^{L_{\max}}.
\end{align*}
By the definition of $L_{\max}$,
\[
  \lambda_\gamma^{L_{\max}}
  =
  (1+2\gamma)^{-L_{\max}}
  \leq
  \frac{\gamma^4}{2}
  \leq
  \frac{\gamma^2\Delta^2}{2},
\]
where the last inequality uses $\gamma\leq\Delta$. Consequently,
\[
  p_{\mathrm{Search}}^{(<L_{\max})}
  \geq
  \frac12\gamma^2\Delta^2
  \geq
  \frac12\gamma^4.
\]

Finally, since $0<\gamma\leq\Delta\leq1$ and
$\log(1+2\gamma)\geq 2\gamma/(1+2\gamma)\geq2\gamma/3$,
\[
  L_{\max}
  =
  O\!\left(\gamma^{-1}\log(2/\gamma)\right)
  =
  O(\gamma^{-2}).
\]
Thus, the truncated procedure is uniform, and uses at most polynomially many masks and decoder-reflection applications. In particular, it has uniform worst-case polynomial in $n$ running time whenever $1/\gamma$ and the running time of the original strategy are polynomial in $n$.
\end{proof}

Finally, we show that it is not necessary to know any lower bound on
$\Delta$, thereby making the expected-time reduction completely uniform.

\begin{theorem}[Variant when a lower bound on $\Delta$ is not known]
\label{thm:unknown-delta}
Let $\Delta>0$. There exists a uniform reduction that does not take
$\Delta$, or any lower bound on $\Delta$, as input. Its expected running
time is linear in the running time of the strategy for
$\GL(G)$ (and independent of $\Delta$), and the resulting strategy for $G$ succeeds with probability at least
\begin{equation*}
    \frac{3}{16}\Delta^6.
    \label{eq:unknown-delta-success}
\end{equation*}
\end{theorem}

\begin{proof}
The reduction first samples an integer $I\geq1$ according to
\begin{equation}
    \Pr[I=i]=3\cdot 4^{-i},
    \qquad i=1,2,\ldots,
    \label{eq:unknown-delta-distribution}
\end{equation}
sets $\gamma=2^{-I}$, and then runs the procedure in
Reduction~\ref{alg:gl-reduction} with this value of $\gamma$. The
procedure is well defined for every $\gamma>0$; the condition
$\gamma\leq\Delta$ is needed only for its success guarantee.

The distribution in \eqref{eq:unknown-delta-distribution} can be sampled
uniformly as follows: repeatedly sample two bits, and stop at the
first pair that is not $00$. If this occurs on the $i$th trial, set
$I=i$.

Let $j\geq1$ be such that
\begin{equation*}
    \frac{\Delta}{2}\leq 2^{-j}\leq\Delta.
    \label{eq:matched-dyadic-scale}
\end{equation*}
Such a $j$ exists for every $0<\Delta\leq1$ and is used only in the
analysis; the reduction does not need to know it. Conditioned on the
event $I=j$, Theorem~\ref{thm:parametric-reduction} applies and gives
success probability at least $(2^{-j}\Delta)^2$. Ignoring all other
branches, which contribute nonnegative success probability, we obtain
\begin{align*}
    p_{\mathrm{Search}}
    &\geq
    \Pr[I=j]\,(2^{-j}\Delta)^2 \\
    &=
    3\cdot4^{-j}\,(2^{-j}\Delta)^2 \\
    &=
    3(2^{-j})^4\Delta^2 \\
    &\geq
    3\left(\frac{\Delta}{2}\right)^4\Delta^2
    =
    \frac{3}{16}\Delta^6.
\end{align*}

It remains to bound the expected running time. Conditioned on $I=i$, we
have $\gamma=2^{-i}$, and hence the expected sequence length in
Reduction~\ref{alg:gl-reduction} is
\begin{equation*}
    \E[L\mid I=i]
    =
    \frac{1}{2\gamma}
    =
    2^{i-1}.
\end{equation*}
Therefore,
\begin{equation*}
    \E[L]
    =
    \sum_{i=1}^{\infty}
    3\cdot4^{-i}\,2^{i-1} =
    \frac{3}{2}
    \sum_{i=1}^{\infty}2^{-i}
    =
    \frac{3}{2}.
\end{equation*}
Thus, the expected number of sampled masks and decoder-reflection calls
is constant, and the overall expected running time is linear in the running time of the original strategy for $\GL(G)$.
\end{proof}

\printbibliography

@inproceedings{GL89,
  author    = {O. Goldreich and L. A. Levin},
  title     = {A Hard-Core Predicate for All One-Way Functions},
  booktitle = {Proceedings of the Twenty-First Annual ACM Symposium on Theory of Computing ({STOC} '89)},
  pages     = {25--32},
  publisher = {ACM Press},
  year      = {1989},
  doi       = {10.1145/73007.73010},
  sortkey   = {01}
}

@article{bhattacharyya2026statistically,
  title={Statistically secure uncloneable encryption of arbitrary messages},
  author={Bhattacharyya, Archishna and Broadbent, Anne and Culf, Eric},
  journal={arXiv preprint arXiv:2607.28561},
  year={2026}
}

@misc{BDF26,
  author       = {Pierre Botteron and S\'ebastien Designolle and Omar Fawzi},
  title        = {{Pauli} Encodings \& Unclonable Encryption},
  year         = {2026},
  note         = {{arXiv}:2607.28602},
  doi          = {10.48550/arXiv.2607.28602},
  url          = {https://arxiv.org/abs/2607.28602},
  sortkey      = {22}
}

@article{BotteronEtAl26,
  author       = {Pierre Botteron and Anne Broadbent and Eric Culf and
                  Ion Nechita and Cl\'ement Pellegrini and Denis Rochette},
  title        = {Towards Unconditional Uncloneable Encryption},
  journal      = {Quantum},
  volume       = {10},
  pages        = {2157},
  year         = {2026},
  doi          = {10.22331/q-2026-07-08-2157},
  url          = {https://quantum-journal.org/papers/q-2026-07-08-2157/},
  sortkey      = {17}
}

@inproceedings{AC02,
  author    = {M. Adcock and R. Cleve},
  title     = {A Quantum {Goldreich--Levin} Theorem with Cryptographic Applications},
  editor    = {H. Alt and A. Ferreira},
  booktitle = {{STACS} 2002},
  series    = {Lecture Notes in Computer Science},
  volume    = {2285},
  pages     = {323--334},
  publisher = {Springer},
  year      = {2002},
  doi       = {10.1007/3-540-45841-7_26},
  sortkey   = {02}
}

@inproceedings{AKL23,
  author    = {P. Ananth and F. Kaleoglu and Q. Liu},
  title     = {Cloning Games: A General Framework for Unclonable Primitives},
  editor    = {H. Handschuh and A. Lysyanskaya},
  booktitle = {Advances in Cryptology---{CRYPTO} 2023, Part V},
  series    = {Lecture Notes in Computer Science},
  volume    = {14085},
  pages     = {66--98},
  publisher = {Springer},
  year      = {2023},
  doi       = {10.1007/978-3-031-38554-4_3},
  sortkey   = {03}
}

@inproceedings{CLLZ21,
  author    = {A. Coladangelo and J. Liu and Q. Liu and M. Zhandry},
  title     = {Hidden Cosets and Applications to Unclonable Cryptography},
  editor    = {T. Malkin and C. Peikert},
  booktitle = {Advances in Cryptology---{CRYPTO} 2021, Part I},
  series    = {Lecture Notes in Computer Science},
  volume    = {12825},
  pages     = {556--584},
  publisher = {Springer},
  year      = {2021},
  doi       = {10.1007/978-3-030-84242-0_20},
  sortkey   = {04}
}

@article{Got03,
  author       = {Daniel Gottesman},
  title        = {Uncloneable Encryption},
  journaltitle = {Quantum Information and Computation},
  volume       = {3},
  number       = {6},
  pages        = {581--602},
  year         = {2003},
  doi          = {10.26421/QIC3.6-2},
  sortkey      = {05}
}

@article{CV22,
  author       = {Eric Culf and Thomas Vidick},
  title        = {A Monogamy-of-Entanglement Game for Subspace Coset States},
  journaltitle = {Quantum},
  volume       = {6},
  pages        = {791},
  year         = {2022},
  doi          = {10.22331/q-2022-09-01-791},
  sortkey      = {04a}
}

@inproceedings{BL20,
  author    = {Anne Broadbent and S\'ebastien Lord},
  title     = {Uncloneable Quantum Encryption via Oracles},
  editor    = {Steven T. Flammia},
  booktitle = {15th Conference on the Theory of Quantum Computation,
               Communication and Cryptography ({TQC} 2020)},
  series    = {Leibniz International Proceedings in Informatics ({LIPIcs})},
  volume    = {158},
  pages     = {4:1--4:22},
  publisher = {Schloss Dagstuhl--Leibniz-Zentrum f\"ur Informatik},
  year      = {2020},
  doi       = {10.4230/LIPIcs.TQC.2020.4},
  sortkey   = {06}
}

@article{TFKW13,
  author       = {Marco Tomamichel and Serge Fehr and J\k{e}drzej Kaniewski and
                  Stephanie Wehner},
  title        = {A Monogamy-of-Entanglement Game with Applications to
                  Device-Independent Quantum Cryptography},
  journaltitle = {New Journal of Physics},
  volume       = {15},
  number       = {10},
  pages        = {103002},
  year         = {2013},
  doi          = {10.1088/1367-2630/15/10/103002},
  sortkey      = {07}
}

@inproceedings{AKLLZ22,
  author    = {Prabhanjan Ananth and Fatih Kaleoglu and Xingjian Li and
               Qipeng Liu and Mark Zhandry},
  title     = {On the Feasibility of Unclonable Encryption, and More},
  editor    = {Yevgeniy Dodis and Thomas Shrimpton},
  booktitle = {Advances in Cryptology---{CRYPTO} 2022, Part II},
  series    = {Lecture Notes in Computer Science},
  volume    = {13508},
  pages     = {212--241},
  publisher = {Springer},
  year      = {2022},
  doi       = {10.1007/978-3-031-15979-4_8},
  sortkey   = {07a}
}

@article{KT25,
  author       = {Srijita Kundu and Ernest Y.-Z. Tan},
  title        = {Device-Independent Uncloneable Encryption},
  journaltitle = {Quantum},
  volume       = {9},
  pages        = {1582},
  year         = {2025},
  doi          = {10.22331/q-2025-01-08-1582},
  sortkey      = {08}
}

@inproceedings{AKY25,
  author    = {Prabhanjan Ananth and Fatih Kaleoglu and Henry Yuen},
  title     = {Simultaneous {Haar} Indistinguishability with Applications to
               Unclonable Cryptography},
  editor    = {Raghu Meka},
  booktitle = {16th Innovations in Theoretical Computer Science Conference
               ({ITCS} 2025)},
  series    = {Leibniz International Proceedings in Informatics ({LIPIcs})},
  volume    = {325},
  pages     = {7:1--7:23},
  publisher = {Schloss Dagstuhl--Leibniz-Zentrum f\"ur Informatik},
  year      = {2025},
  doi       = {10.4230/LIPIcs.ITCS.2025.7},
  sortkey   = {11}
}

@inproceedings{AB24,
  author    = {Prabhanjan Ananth and Amit Behera},
  title     = {A Modular Approach to Unclonable Cryptography},
  editor    = {Leonid Reyzin and Douglas Stebila},
  booktitle = {Advances in Cryptology---{CRYPTO} 2024, Part VII},
  series    = {Lecture Notes in Computer Science},
  volume    = {14926},
  pages     = {3--37},
  publisher = {Springer},
  year      = {2024},
  doi       = {10.1007/978-3-031-68394-7_1},
  sortkey   = {12}
}

@misc{CHV24,
  author       = {C\'eline Chevalier and Paul Hermouet and Quoc-Huy Vu},
  title        = {Towards Unclonable Cryptography in the Plain Model},
  year         = {2023},
  howpublished = {Cryptology ePrint Archive, Report 2023/1825},
  eprint       = {2311.16663},
  eprinttype   = {arXiv},
  url          = {https://eprint.iacr.org/2023/1825},
  sortkey      = {13}
}

@inproceedings{CLX26,
  author    = {Andrea Coladangelo and Qipeng Liu and Ziyi Xie},
  title     = {The Curious Case of ``{XOR} Repetition'' of
               Monogamy-of-Entanglement Games},
  editor    = {Shubhangi Saraf},
  booktitle = {17th Innovations in Theoretical Computer Science Conference
               ({ITCS} 2026)},
  series    = {Leibniz International Proceedings in Informatics ({LIPIcs})},
  volume    = {362},
  pages     = {41:1--41:20},
  publisher = {Schloss Dagstuhl--Leibniz-Zentrum f\"ur Informatik},
  year      = {2026},
  doi       = {10.4230/LIPIcs.ITCS.2026.41},
  sortkey   = {14}
}

@inproceedings{CKNY25,
  author    = {Jeffrey Champion and Fuyuki Kitagawa and Ryo Nishimaki and
               Takashi Yamakawa},
  title     = {Untelegraphable Encryption and Its Applications},
  booktitle = {Theory of Cryptography---{TCC} 2025, Part III},
  series    = {Lecture Notes in Computer Science},
  volume    = {16270},
  pages     = {3--35},
  publisher = {Springer},
  year      = {2025},
  doi       = {10.1007/978-3-032-12296-4_1},
  sortkey   = {15}
}

@misc{CakanGoyal24,
  author       = {Alper Cakan and Vipul Goyal},
  title        = {Unbounded Leakage-Resilient Encryption and Signatures},
  year         = {2024},
  howpublished = {Cryptology ePrint Archive, Report 2024/1876},
  url          = {https://eprint.iacr.org/2024/1876},
  sortkey      = {16}
}

@article{BC26,
  author       = {Archishna Bhattacharyya and Eric Culf},
  title        = {Uncloneable Encryption from Decoupling},
  journaltitle = {Nature Physics},
  volume       = {22},
  number       = {2},
  pages        = {315--318},
  year         = {2026},
  doi          = {10.1038/s41567-025-03154-7},
  sortkey      = {18}
}

@misc{BBC26,
  author       = {Archishna Bhattacharyya and Anne Broadbent and Eric Culf},
  title        = {The Uncloneable Bit Exists},
  year         = {2026},
  eprint       = {2603.08916},
  eprinttype   = {arXiv},
  url          = {https://arxiv.org/abs/2603.08916},
  sortkey      = {19}
}

@misc{AS26,
  author       = {Prabhanjan Ananth and Amit Sahai},
  title        = {Unconditional Unclonable Encryption},
  year         = {2026},
  howpublished = {Cryptology ePrint Archive, Report 2026/1511},
  eprint       = {2607.21551},
  eprinttype   = {arXiv},
  url          = {https://eprint.iacr.org/2026/1511},
  sortkey      = {20}
}

@misc{Rag26,
  author       = {Seyoon Ragavan},
  title        = {Efficient Unclonable Encryption from {Pauli} Eigenstates},
  year         = {2026},
  howpublished = {Cryptology ePrint Archive, Report 2026/1509},
  eprint       = {2607.21811},
  eprinttype   = {arXiv},
  url          = {https://eprint.iacr.org/2026/1509},
  sortkey      = {21}
}

@inproceedings{anshu2022area,
  title={An area law for 2d frustration-free spin systems},
  author={Anshu, Anurag and Arad, Itai and Gosset, David},
  booktitle={Proceedings of the 54th Annual ACM SIGACT Symposium on Theory of Computing},
  pages={12--18},
  year={2022}
}

@article{arad2012improved,
  title={Improved one-dimensional area law for frustration-free systems},
  author={Arad, Itai and Landau, Zeph and Vazirani, Umesh},
  journal={Physical Review B—Condensed Matter and Materials Physics},
  volume={85},
  number={19},
  pages={195145},
  year={2012},
  publisher={APS}
}

@inproceedings{gilyen2019quantum,
  title={Quantum singular value transformation and beyond: exponential improvements for quantum matrix arithmetics},
  author={Gily{\'e}n, Andr{\'a}s and Su, Yuan and Low, Guang Hao and Wiebe, Nathan},
  booktitle={Proceedings of the 51st annual ACM SIGACT symposium on theory of computing},
  pages={193--204},
  year={2019}
}

@article{ananth2026non,
  title={Non-Local Search-to-Decision Reduction over F2},
  author={Ananth, Prabhanjan},
  journal={arXiv preprint arXiv:2608.19091},
  year={2026}
}

\appendix

\section{A common-mask strategy on which plain GL fails completely}
\label{app:cancellation}
We describe a simple two-bit cloning game and a strategy for its common-GL mask game whose prediction probability is strictly larger than $1/2$, while the corresponding plain Goldreich--Levin simultaneous extraction probability is zero (and hence the overlap between vectors $\ket{b_{k,m}}, \ket{c_{k,m}}$ is also zero).

We consider the following distribution over $(k,m)$: uniform over all pairs satisfying $k=m$. Although $k$ completely reveals the message $m$ in this setting, it is still interesting to observe that plain GL fails.

The challenger prepares the following fixed quantum state (regardless  of $k,m$) and Alice does nothing:
\[
  \ket{\psi}^{A_1A_2}
  :=
  \frac{1}{\sqrt6}
  \left(
    \ket{00}+\ket{01}+2\ket{11}
  \right).
\]
Consider a splitting procedure that sends $A_1$ to Bob and $A_2$ to Charlie. Thus, the
post-splitting state is exactly $\ket{\psi}^{BC}$. For masks $r\in\{00,01,10,11\}$, define the local reflections
\[
\begin{array}{c|cccc}
  r       & 00 & 01 & 10 & 11 \\ \hline
  B_r     & Z  & X  & \Id & -X \\
  C_r     & Z  & -Z & \Id & -X
\end{array}
\]
and let $\chi_x(r):=(-1)^{\langle r,x\rangle}$. Then, define Bob's and Charlie's decoder reflections to be
\[
  V^B_{k,r}:=\chi_k(r)B_r,
  \qquad
  V^C_{k,r}:=\chi_k(r)C_r.
\]
Since $k=m$, note that the reflections about the respective correct-output subspaces are exactly
\[
  \chi_m(r)V^B_{k,r}=B_r,
  \qquad
  \chi_m(r)V^C_{k,r}=C_r,
\]
which act as $+1$ on the correct-output subspace and $-1$ on incorrect.
Then, the four projectors corresponding to simultaneous correct guesses, for $r=00,01,10,11$ respectively, are
\[
  \proj{0}\otimes\proj{0},
  \qquad
  \proj{+}\otimes\proj{1},
  \qquad
  \Id,
  \qquad
  \proj{-}\otimes\proj{-}.
\]
Their expectations with respect to $\ket{\psi}$ can be computed to be
\[
  \frac16,
  \qquad
  \frac34,
  \qquad
  1,
  \qquad
  \frac16.
\]
So, the common-mask prediction probability is
\[
  p_{\mathrm{Pred}}
  =
  \frac14
  \left(
    \frac16+\frac34+1+\frac16
  \right)
  =
  \frac{25}{48}
  >
  \frac12.
\]

Now, consider what happens when Bob and Charlie each apply their plain GL extractors associated with the above strategy. Since $k=m$, the relevant Fourier coefficients are
\begin{align*}
  \widehat V^B_k(m)
  &=
  \frac14\sum_r \chi_m(r)V^B_{k,r}
  =
  \frac14\sum_r B_r
  =
  \frac{\Id+Z}{4}
  =
  \frac12\proj{0},\\
  \widehat V^C_k(m)
  &=
  \frac14\sum_r \chi_m(r)V^C_{k,r}
  =
  \frac14\sum_r C_r
  =
  \frac{\Id-X}{4}
  =
  \frac12\proj{-}.
\end{align*}
However, $
  (\bra{0}\bra{-})\ket{\psi}
  =
  \frac{1}{\sqrt2}
  \left(
    \bra{00}-\bra{01}
  \right)\ket{\psi}
  =
  0$, which implies that the entire simultaneous correct-output branch vanishes:
\[
  \left(
    \widehat V^B_k(m)\otimes
    \widehat V^C_k(m)
  \right)
  \ket{\psi}
  =
  0.
\]
Therefore, the plain GL extraction probability is exactly
\[
  p_{\mathrm{plainGL}}
  =
  \left\|
    \left(
      \widehat V^B_k(m)\otimes
      \widehat V^C_k(m)
    \right)
    \ket{\psi}
  \right\|^2
  =
  0.
\]
One can also check that the individual correct-output vectors $
  \ket{b_{k,m}}
  =
  \left(
    \widehat V^B_k(m)\otimes\Id
  \right)\ket{\psi}$ and $\ket{c_{k,m}}
  =
  \left(
    \Id\otimes\widehat V^C_k(m)
  \right)\ket{\psi}$, satisfy
\[
  \braket{b_{k,m}}{c_{k,m}}
  =
  \bra{\psi}
  \left(
    \widehat V^B_k(m)\otimes
    \widehat V^C_k(m)
  \right)
  \ket{\psi}
  =
  0.
\]
Thus, $p_{\mathrm{Pred}}=\frac{25}{48}>\frac12$, but $p_{\mathrm{plainGL}} = \braket{b_{k,m}}{c_{k,m}}=0$.
\end{document}